\documentclass[11pt]{article}
\usepackage[utf8]{inputenc}
 \usepackage{geometry}
\usepackage{amsmath,amsfonts,amsthm,amssymb}
\usepackage{algorithm,algorithmic}
\usepackage{graphicx}
\usepackage{textcomp}
\usepackage{xcolor}
\usepackage{todonotes} 

\newtheorem{theorem}{Theorem}
\newtheorem{lemma}{Lemma}
\usepackage[english]{babel}

\usepackage{makecell}
\usepackage{subcaption}
\usepackage{graphicx}
\DeclareMathOperator*{\argmax}{argmax}
\usepackage{enumitem}
\usepackage{cases}
\usepackage{tikz}
\usetikzlibrary{shapes,arrows,positioning}
\def\hN{\hat{N}}
\def\hS{\overline{S}}
\def\hp{\overline{p}}

\title{ 
Equilibrium and Selfish Behavior in Network Contagion
}

\author{Yi Zhang \footnote{PhD student, Illinois Tech} and Sanjiv Kapoor\footnote{Authors  contributed equally}}

\date{}

\begin{document}
\maketitle
\begin{abstract}
In this paper, we consider contagion policy games,  modelled as a non-atomic game in populations that are provided with a choice of preventive policies to act against a contagion spreading amongst interacting populations, be it connected computing devices or biological organisms. The dynamics of the contagion spread  are determined by the standard SIR model.
Each participant of the population has a choice from amongst a set of precautionary policies  with each policy presenting a payoff or utility, which we assume is the same within each group,
the risk being the possibility of infection. The policy groups interact with each other.
We define a network model to model interactions between different population sets. The population sets reside at nodes of the network and follow policies available at that node. We define game-theoretic models  and study  the computational complexity of computing Nash equilibrium and the inefficiency of allowing for individual decision making, as opposed to centralized control. 

We show that
computing Nash equilibrium for interacting policy groups is, in general, PPAD-hard.   For the case of policy groups where the interaction is uniform for each group, i.e., each group's impact is dependent on the characteristics of that group's policy, we present a polynomial-time algorithm for computing Nash equilibrium. This requires us to compute the size of the susceptible set at the endemic state and investigating the computational complexity of computing endemic equilibrium is of importance. We present a convex program to compute the endemic equilibrium for $n$ interacting policies in polynomial time. This leads us to investigate equilibrium in  the network model and we present an algorithm to compute Nash equilibrium policies at each node of the network assuming uniform interaction.

We also analyze the price of anarchy considering the multiple scenarios.
For an interacting population model, we determine that the price of anarchy is  a constant $e^{R_0}/R_0$ where $R_0=\beta/\gamma$ is the reproduction number of the contagion, a parameter that reflects the growth rate of the virus. As an example, $R_0$ for the original COVID-19 virus was estimated to be between 1.4 and 2.4 by WHO at the start of the pandemic. The relationship of PoA to the contagion's reproduction number is surprising.
 Network models that capture the interaction of distinct population sets (e.g. distinct countries) have a  PoA that again grows exponentially with $R_0$ with multiplicative factors that reflect the interaction between population groups across nodes in the network.
 

\end{abstract}

\setcounter{tocdepth}{2} 


\section{Introduction}
In this paper, we consider  contagion  games with preventative policy control in networks  where the contagion spread is modeled by an SIR process. Contagion could be computer virus spread, biodisease spread or misinformation spread etc. We assume that the  population chooses  a preventative policy, from amongst a set of control policies suggested by an administrator (e.g. government). Each policy impacts the economic well-being or pay-off of the group that adopts the policy. Examples of policies during infectious disease pandemics, like the COVID-19 pandemic, include pharmaceutical and non-pharmaceutical policies (NPI)  ranging from vaccination in the former to masking, stay-at-home, and hand washing, etc. in the latter. While policies were sometimes mandated, more often they were recommended. As such, policy adoption was guided by the self-interest of each constituent of the population.
In general, we consider $n$ different individual policies or combinations thereof. The end result of most contagions is the achievement of endemic equilibrium when
there is no more growth of the virus. This typically happens when the recovery rate is more than the rate of infections, thus leading to a steady state, reducing the infectious population to zero. A similar situation occurs in other contexts, including computer viruses.

One standard model for contagion processes is the SIR compartment model\cite{kermack1927contribution} that utilizes the susceptible (S), infectious (I), and removed (R)(recovered or dead) set of populations. 
With each policy adopted, there is an associated pay-off along with a risk of infections that is indicated by the transmission factor of the contagion when following the policy.
This defines a non-atomic game, termed the contagion game, where the population elects a particular policy based on utility functions.
The utility function is defined  as an increasing function of the population size $S$ at the final state of the contagion spread, i.e. the final size of the population that evades the contagion upon following a specific policy, based on the benefit that the individual receives from following the policy. 
In this paper, we consider the computation of Nash equilibrium and the price of anarchy that is a consequence of selfish policy choices. We consider a social utility that is the sum of the pay-offs to each policy group. This requires us to bound the size of the susceptible group at the endemic state. 

In this paper, we consider (i) a model where the population interacts with each other, modeled by a complete graph of interaction over the policy groups (ii) a model of sets of interacting populations over a graph or network where each node represents a population. e.g. a country or a company, with a set of preventive policies at each node, and edges represent the interaction of two different population sets.  
We show that computing the Nash equilibrium in contagion games with $n$ general policies is PPAD-hard, even when restricted to one node in the network. Nash equilibrium in this game is considerably difficult to characterize since the interactions 
between policy groups are arbitrary.
 We thus consider two scenarios, one where interactions between nodes and groups following different policies are restricted to be uniform, i.e. dependent on the node and group only, and the simpler case where the population following each policy does not interact with the other populations.
 We show that determining 
 Nash equilibrium in these models is of polynomial complexity. The price of anarchy in the two models is shown to be exponentially related to $R_0$, the reproduction number associated with the virus.

Of independent interest is the problem of computing final sizes\cite{magal2016final} of the susceptible population in an SIR model. 
This is a problem arising in multiple fields including biology and mathematical physics. 
While there is extensive research on SIR models, analytic solutions to these models are not known. 
We present algorithms to compute the size of the susceptible population at endemic equilibrium in the two models we consider.
For the separable policy model, this is the computation of a single variable fixed point of a function. For the interacting policy models, including the network model, the computation of endemic equilibrium involves computations of fixed points of multi-variate functions, due to interaction between subsets of population that follow different policies. 
For this case, we present a convex program to compute the  equilibrium for $n$ interacting policies

Game-theoretic formulations have been used in the study of policies that attempt to contain the spread of contagion. 
Multiple applications of game theory in the context of contagion, malicious players  or pandemic infection spread may be found discussed in \cite{huang2022game,chang2020game,funk2010modelling,EC8-Innoculation,EC07MosheKPapad,GOYAL201564}. 
Simulation-based analysis of selfish behavior in adopting policies have been investigated in \cite{bauch2003group} and a simulated analysis of the price of anarchy may be found in \cite{nicolaides2013price}, the results dealing with mobility-related contagion spread. 
In the paper\cite{nicolaides2013price} the authors study transportation-related spread of contagions through selfish routing strategies as contrasted with policy-suggested routes. Using simulation they show that selfish behavior leads to increase in the total population infected. Malicious players in the context of congestion games have been studied in \cite{EC07MosheKPapad}. A vaccination game in a network setting has been considered in \cite{ASPNES20061077} where each node adopts a strategy to vaccinate or not,  an infection process that starts at a node and spreads across all the subgraph of all unprotected nodes connected to the start node. Defining the social welfare to be the size of the infected set, the {\em price of anarchy} has been shown to be infinity.
 The model considers utilities dependent on the node being infected or not and the
 cost of vaccination.
Distinct from the above, in this paper we consider a non-atomic game on a network of populations where the spread of infection is modeled by a SIR process. We may note that the SIR process has also been used for modeling the spread of false information, as surveyed in  \cite{SIR-False-10.1145/3563388}.

 Vaccination games in the non-atomic setting have been considered in  \cite{bauch2004vaccination}, where the payoffs utilized are morbidity risks from vaccination and infection, the strategies being either to adopt the recommended policy of vaccinations or to ignore the advice and risk the higher probability of infection. The authors
 compute the Nash equilibrium in this two-strategy game and conclude that it is impossible to eradicate an infectious disease through voluntary vaccinations due to the selfish interests of the population. 
 In another investigation of vaccination games\cite{bauch2003group} the strategies of the game are one of (i)  vaccination or (ii) ``vaccination upon infection". The results comparing group interest and selfish behavior again indicate differences and reduced uptake of vaccines under voluntary programs.

Contagion games have been studied in the context of influence in economic models of competition.
The results in \cite{goyal2012competitive,draief2014new} discuss a game-theoretic formulation of competition in a social network, with firms trying to gain consumers, seeding their influence at nodes, with monetary incentives. In this case, the firms would like to
have maximum influence over the nodes. 
This approach has strategies that depend on the dynamics of adoption and budgets of the firms along with the structure of the social network and the authors discuss the price of anarchy in this context. 
Additional work on such games may be found in \cite{easley2010networks,kumar2019choosing,tsai2012security}.

Our model may also apply in this context when considering influence networks where the influence spread is guided by a dynamical system of equations specified by the SIR process when the utility is a function of the non-infected population but susceptible population $S$.
\section{Models and Results: The Contagion Game}

The underlying model for contagion spread is the SIR model, where $S$ is the set of susceptible population, $I$ the set of infectious population and $R$ the removed set, either through recovery or death.
We assume that to prevent the spread of contagion, governments or administrators specify $n$ policies in the set ${\mathcal P}= \{ P_1, P_2, \ldots P_n\}$. 
Example policies in the context of the COVID-19 pandemic could be the adoption of masks, shelter-at-home etc. 
We define a non-atomic game where an infinitesimal-sized player decides to follow one of the $n$ policies. Normalizing the population to be of unit size, let $\phi_i\geq 0$ be the fraction of the population that follows policy $P_i$, with the total population $\sum_i \phi_i=1$. 
The SIR process is identified by a set of differential equations that govern the movement of population between compartments representing $S$, $I$ and $R$ and ends at time $T=\infty$, when endemic equilibrium is achieved. The infectious group $I_i(\infty)$ drops to 0 and the size of the susceptible population converges to $S_i(\infty)$. 
In this paper, we refer to $S_i(\infty)$ as the {\bf final size}, which is important to the group's payoff and often denote it by $S_i$ for simplicity in later analysis.
Each group's total population is closed. With $I_i(\infty)=0$, we get 
$S_i(\infty) + R_i(\infty) = \phi_i$.
\subsection{Models}
We describe our models starting with the simple models of interacting policy groups, where the interaction between populations following distinct policies is complete. We then consider a network model where each node represents a distinct population following a set of policies specific to the node and the susceptible population at a node interacts with other nodes as well as amongst itself.\\ \\
\noindent
{\bf{Model A: Interacting Policy Sets:}\label{SIR_model1}}\\

\begin{figure}
    \centering
    \begin{tikzpicture}[->,>=stealth',shorten >=2pt, line width=1pt,style ={minimum size=10mm}]
    \node[rectangle, draw](S1) at (0.5,0){$S_1$};
    \node[rectangle, draw](I1) at (4,0){$I_1$};
    \node[rectangle, draw](R1) at (6,0){$R_1$};
    \node[rectangle, draw](S2) at (0.5,-1.5){$S_2$};
    \node[rectangle, draw](I2) at (4,-1.5){$I_2$};
    \node[rectangle, draw](R2) at (6,-1.5){$R_2$};
    \node[rectangle, draw](S3) at (0.5,-3){$S_3$};
    \node[rectangle, draw](I3) at (4,-3){$I_3$};
    \node[rectangle, draw](R3) at (6,-3){$R_3$};

    \draw[->,transform canvas={xshift=0,yshift=0},color=red] (I1) -- (S1);
    \draw[->,transform canvas={xshift=0,yshift=0},color=red] (I1) -- (S2);
    \draw[->,transform canvas={xshift=0,yshift=0},color=red] (I1) -- (S3);
    \draw[->] (I1) -- (R1);
    \draw[->,transform canvas={xshift=0,yshift=0},color=red] (I2) -- (S1);
    \draw[->,transform canvas={xshift=0,yshift=0},color=red] (I2) -- (S2);
    \draw[->,transform canvas={xshift=0,yshift=0},color=red] (I2) -- (S3);
    \draw[->] (I2) -- (R2);
    \draw[->,transform canvas={xshift=0,yshift=0},color=red] (I3) -- (S1);
    \draw[->,transform canvas={xshift=0,yshift=0},color=red] (I3) -- (S2);
    \draw[->,transform canvas={xshift=0,yshift=0},color=red] (I3) -- (S3);
    \draw[->,transform canvas={xshift=0,yshift=3}] (S1) -- (I1);
    \draw[->,transform canvas={xshift=0,yshift=3}] (S2) -- (I2);
    \draw[->,transform canvas={xshift=0,yshift=3}] (S3) -- (I3);
    \draw[->] (I3) -- (R3);
    \end{tikzpicture}
    \caption{Transition diagram of the interacting model ({\bf Model A}). The red arrows indicate interactions across different groups.}
    \label{fig:interacting}
\end{figure}
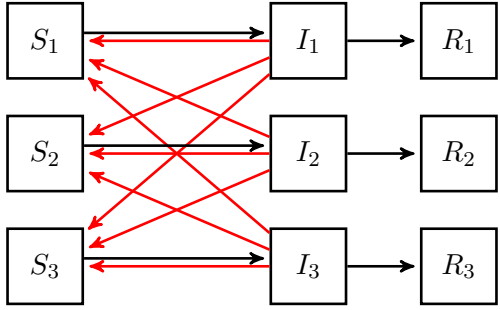\noindent
We first consider the standard interacting model for contagion spread, later modified by uniform interaction policies and separable policies. 

\begin{itemize}
    \item 
{\bf{A1. General Interaction Policy Sets:}}\\
Let $\gamma>0$ be the rate of removal(recovery) of the infectious group. Let $\beta$ be an $n$ by $n$ non-negative matrix of transmissive parameters of infection.
At any time $t$, the dynamics of each compartment of policy group $i$ are defined as follows.
\[
\beta=
\begin{bmatrix}
    \beta_{1,1} & \cdots    & \beta_{1,n}   \\
    \vdots      & \ddots    & \vdots        \\
    \beta_{n,1} & \cdots    & \beta_{n,n}
\end{bmatrix},\quad\quad
\begin{array}{l}
    \frac{\mathrm{d}S_i(t)}{\mathrm{d}t}=
    -S_i(t)\cdot\sum_{j=1}^n \beta_{i,j}I_j(t),\\
    \\
    \frac{\mathrm{d}I_i(t)}{\mathrm{d}t}=
    S_i(t)\cdot\sum_{j=1}^n \beta_{i,j}I_j(t)
    - \gamma I_i(t),\\
    \\
    \frac{\mathrm{d}R_i(t)}{\mathrm{d}t}= \gamma I_i(t),
\end{array}
\]
The initial conditions when $t=0$ are $  S_i(0)=(1-\epsilon)\phi_i$,
$I_i(0)=\epsilon\phi_i$,
$R_i(0)=0 $, with $S_i(0) + I_i(0) + R_i(0) = \phi_i$.
$\epsilon$ is the initial fraction of the infectious population and is understood to be very small.
Group $i$'s susceptible population $S_i$ will interact with infectious $I_j$ from all group $j$ resulting  in infectious $I_i$, eventually leading  to a removed set $R_i$ 
(which represents either recovered or dead ). With a substantial initial size of susceptible and infectious populations, the infection numbers will typically peak and subsequently reduce.


\item
{\bf{A2. Uniform Interaction Policy Sets:}}\\
In this model, for each group pair $i,j$ we define $\beta_{i,j}=\kappa_i \kappa_j \beta_{0}$, with $1=\kappa_1>\kappa_2>\cdots>\kappa_n>0$.
The interaction between each pair of groups can be decomposed into a product of groups. This simplification assumes that the transmission of a virus between two groups would impact the population of each group equally. So, if one group, $i$, has a very restrictive policy, e.g. ``isolation" where $\kappa_i=0$, then that policy would render the transmission rate with the other group to be zero, applicable to both groups.
We further require the largest reproduction number $R_0=\frac{\beta_0}{\gamma}\geq 1$, for otherwise even if every player joins group 1 which has the highest $\beta$'s, the infection will end immediately.\\

\end{itemize}
{\bf{Model B: Network Interaction Policy Sets:}}\\
\begin{figure}
    \centering
    \includegraphics[width = 0.8\textwidth]{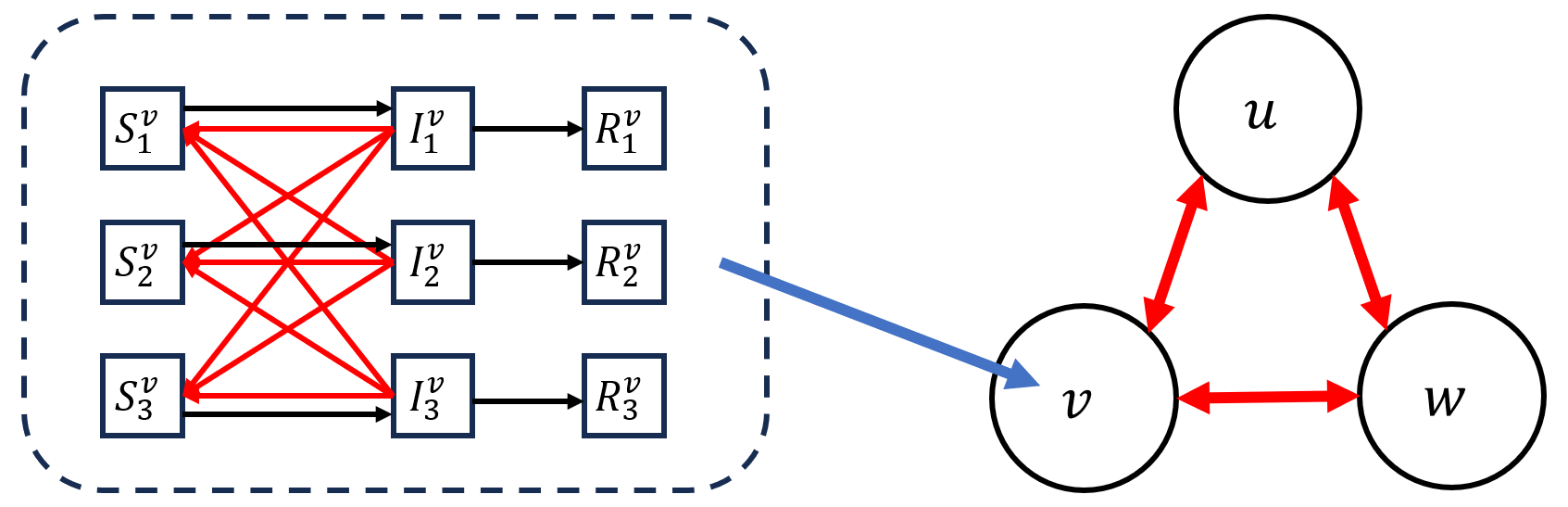}
    \caption{Network interaction model ({\bf Model B}). The red arrows indicate interactions across different groups and different nodes.}
    \label{fig:enter-label}
\end{figure}
\noindent In this model we extend the interacting policy sets to a network over $m$ nodes. Each node contains $n$ policy groups. For simplicity we assume a common set of policies, ${\mathcal{P}=\{ P_1, P_2, \ldots P_n\}}$, over the network. In the entire network, a group is identified by the pair $(v,i)$, where $v$ is the node it is located in and $i$ is its group number inside $v$. The initial population of each group $(v,i)$ is denoted by $\phi_i^v$. Each node's initial population sums up to 1, i.e., $\sum_i \phi_i^v=1,\forall v$. The virus transmission parameter between any group pair $(v,i),(u,j)$ is defined as $\beta_{i,j}^{v,u} = \alpha_{v,u}\beta_{i,j}$, where $\alpha_{v,u}$ defines the interaction between node $v$ and $u$.
Each node $v$ has a strategy profile $\phi^v = [\phi_1^v,\cdots,\phi_n^v]^T\geq 0$ with $\sum_i \phi_i^v = 1$ that represents the population following policy $P_i$.  
Given $\phi_i^v$, the population of group $i$ at node $v$, we define the initial conditions to be $S_i^v(0) = (1-\epsilon) \phi_i^v$, $I_i^v(0) = \epsilon \phi_i^v$, $R_i^v(0)=0$.
The SIR process is as follows
\[
\begin{array}{l}
    \frac{\mathrm{d}S_i^v(t)}{\mathrm{d}t}=
    -S_i^v(t)\cdot\sum_{(u,j)} \beta_{i,j}^{v,u}I_j^u(t),\\
    \\
    \frac{\mathrm{d}I_i^v(t)}{\mathrm{d}t}=
    S_i^v(t)\cdot\sum_{(u,j)} \beta_{i,j}^{v,u}I_j^u(t)
    - \gamma I_i^v(t),\\
    \\
    \frac{\mathrm{d}R_i^v(t)}{\mathrm{d}t}= \gamma I_i^v(t),
\end{array}
\]

Recall that $\beta_{i,j}^{v,u}$ is $ \alpha_{v,u}\beta_{i,j}$ for all $(v,i)$ and $(u,j)$. We consider different versions of the game depending on 
$\alpha$ and $\beta$.
\begin{itemize}
    \item{\bf B1. Arbitrary $\beta$:}
    This is the most general model.
    In this case, regardless of $\alpha$, {\bf Model A1}  is a special case of the network version where there is only 1 node in the network. We will show that Nash equilibrium in this model is hard to compute.

    \item{\bf B2. Uniform Node Interaction and Arbitrary Network Interaction:}

    In this case, $\beta$ is uniform , i.e.  $\beta_{i,j} = \kappa_i \kappa_j \beta_0$, for all $i,j$ pair. The interaction $\alpha_{v,U}$ between node $v$ and $u$ is still arbitrary. The hardness of this case remains open.

    \item{\bf B3. Uniform Network Interaction Model:}

    In this model we assume $\beta_{i,j} = \kappa_i \kappa_j \beta_0$, for all $i,j$ pair, with $1 =\kappa_1 >\kappa_2 >\cdots >\kappa_n$. 
    We also assume $\alpha_{v,u} = \alpha_v \alpha_u$, for all $v,u$ node pair, where $\forall u, 0 \leq \alpha_u \leq 1$. 
    For convenience, we denote $\overline{\kappa}_i^v = \alpha_v \kappa_i$, for all group $(v,i)$.
    The {\em interaction factor} of any node $v$ with other nodes is defined as $\omega = \sum_{u} \alpha_u$. It represents the interaction of an element of a population with other populations in the network and would typically  be a constant.
\end{itemize}


\noindent
{\bf{Game Strategies  and Utilities:}}\\
In the non-atomic game for {\bf Model A}, each infinitesimal player's individual strategy set is assumed to be $\mathcal{P}$, representing the policies. 

We next define the utility of adopting policy $P_i \in {\mathcal P}$. The utility functions are assumed to belong to the class ${\mathcal U}$  defined over $\mathcal{R}$
and are real-valued, concave, increasing and invertible. 
We additionally assume that the inverses are proportional, i.e. for every function  $h_i $ belonging to the class ${\mathcal U}$ ,
there exists a constants $c_i$ with the following property, $c_ih_i^{-1}(x) = c_jh_j^{-1}(x)$. A class of functions that satisfy this property can be constructed based on
homogeneous functions of degree $d$, $0\leq d\leq 1$ and payoff vectors as follows: 
Let $p=[p_1,\cdots, p_n]^T$ be a non-negative vector, representing payments. 
At time $t$, denote $\overline{S}_i(t)= \frac{S_i(t)}{\phi_i}$.  We will let $U_i(\hS_i(t))= h_i(\hS_i(t))=$
 $=p_i\cdot h(\overline{S}_i(t))$, where $h()$ is a homogenous function of degree $d$, 
representing the benefit per unit of daily work when a person is not infected.
The factor $p_i$ is interpreted as payment for  the work provided by the average population. As an example, the population working from home will have a payoff which is different from the population in a factory or office environment. For simplicity we will use homogenous function based utility for the remainder of the paper.

When endemic equilibrium is reached, the expected daily individual utility of group $i$ converges to $U_i= h_i(\overline{S}(\infty))$.
Group $i$'s total group utility is the endemic individual utility multiplied by the group size, i.e. $UG_i=\phi_i U_i$.
Each individual player evaluates the expected individual utility for different groups and decides which group to join, forming all the group sizes, with $\sum_{i=1}^n \phi_i=1$. 
Non-atomic Nash equilibrium $\mathcal{N}$, with $\phi^\mathcal{N}= [\phi^\mathcal{N}_1, \cdots, \phi^\mathcal{N}_n]^T$, is achieved when the player chooses to initially join group $i$ only if its individual utility is the highest, i.e. $\forall i|\phi^\mathcal{N}_i>0
\implies U_i \geq U_j, \forall j$.
We call a group $i$ participating in the Nash equilibrium if $\phi_i^\mathcal{N} > 0$.
If there are multiple groups participating in the Nash equilibrium, they must have the same highest individual utility.
We assume the individual utility at Nash equilibrium is always positive.\\

\noindent
{\bf{Price of Anarchy(POA):}}\\
Each game instance, $G$, is defined by payoff function $h$ and the vector $p$.
In each instance, the social welfare is defined to be the summation of all groups' group utility $\sum_{i=1}^n UG_i$.
Each group utility $UG_i$ is a function of all the group size $\phi=[\phi_1,\cdots,\phi_n]^T$. 
Thus the social welfare, $\sum_i UG_i(\phi)$, is a function of $\phi$. The social optimum is
$OPT = \max_\phi \sum_i UG_i(\phi)$.
Denote by $\mathcal{NE}(G)$ to be set of all Nash equilibria of the game $G$.
We define the price of anarchy($POA$) as follows.
\[
POA = \max_{p,h}
\frac{OPT}{\min_{ \mathcal{N} \in \mathcal{NE}(G)} \sum_i UG_i(\phi^\mathcal{N})},
\]
which is the highest ratio of social optimum versus the lowest social welfare of Nash equilibrium among any game instance.

The non-atomic games represented by {\bf Model B} has infinitesimal players at each node with strategies and utilities, similar to {\bf Model A}. 
We let $\phi^v = (\phi^v_1 , \phi^v_2 \ldots \phi^v_n)$ be the strategy profile at node $v$ with $\sum_{i=1}^n \phi_i^v=1$.
We let the utility functions belong to $\mathcal{U}$, the set of functions defined above that
are invertible, concave and increasing.
For each node $v$, let $h^v \in {\mathcal U}$.
Let $p^v=[p_1^v,\cdots,p_n^v]^T$ be a non-negative vector.
For each group $(v,i)$, let $U_i^v=h_i^v(\hS_i^v(\infty))=p_i^v\cdot h^v(\overline{S}_i^v(\infty))$ be its individual utility, and $UG_i^v = \phi_i^v U_i^v$ be its group utility. Note that each $U_i^v$ and $UG_i^v$ is also a function of $\phi$, where $\phi$ represents the strategies of all groups at all nodes. The social welfare function used in this model is $\sum_{v}\sum_i UG_i^v(\phi)$.

\subsection*{Results and Techniques:}
\begin{itemize}
    \item 
    We show that in {\bf Model A1}, while Nash equilibrium exists ({\bf Theorem \ref{thm:NE_existence}}), computing the Nash equilibrium for contagion games with general interacting policy sets is PPAD-hard ({\bf Theorem \ref{thm:PPAD}}). This is not surprising but nevertheless needs to be proved. A similar result holds for {\bf Model B1} with an arbitrary form of $\beta$.
    \item
    For contagion games with $n$ uniform interaction policy sets ({\bf Model A2}) we provide a convex program to determine the final size $S_i(\infty)$ ({\bf Theorem \ref{thm:decomp_convex_program}}). Determining the final size is key to the Nash computations. 
    \item 
    We provide an algorithm to compute the Nash equilibrium for the mode with uniform interaction policy sets ({\bf Model A2}) with complexity $O(n^2(n+\log( 1/\delta))$ where $\delta$ has a polynomial bound in terms of the input size  ({\bf Theorem \ref{thm:decomp_Nash_algo}}). The method utilizes a proof that the computation of Nash equilibrium in a game with $n$ policies can be determined from a polytope that lies in the positive orthant  with two additional constraints.

    \item 
    We provide an algorithm to compute the Nash equilibrium for the network interaction model with uniform interaction policy sets ({\bf Model B3}) with polynomial complexity ({\bf Theorem~\ref{thm:net_Nash_algo}}). In the network model, the space of solutions is exponential in the network size. We reduce this space by establishing a dominance relation among utilities modeled by an acyclic tournament graph, the source node of which provides a potential solution. Polynomial number of graphs are used to determine the Nash solution. We have not found previous usage of this technique.

    \item
    We show that the upper bound of price of anarchy(POA) in the game with uniform interaction policy sets ({\bf Model A2}) is bounded above by $e^{R_0}/R_0$ ({\bf Theorem~\ref{thm:decomp_POA}}) and
   in the uniform network interaction ({\bf Model B3}) is bounded by  $e^{\alpha_{max} \overline{R_0}}$ where $\overline{R}_0 = \omega R_0 $ where $\omega$ is an interaction factor defined over the network, and $\alpha_{max}=\max_u \alpha_u$.
    This is bounded above by $e^{m R_0}   $ for the worst-case value of the interaction factor.
    ({\bf Theorem~\ref{thm:network_POA}}).
    We utilize a monotone property of the final size w.r.t. increase in group size. Simulations show that these  results are not tight and future work could improve these bounds.
     
\end{itemize}
All the algorithms for computing equilibrium determine approximate solutions. Some of the proofs are included in the appendix.

\section{Hardness of Nash Equilibrium in Interacting Policy Sets}\label{sec:PPAD}
\subsection{Existence of Nash Equilibrium}\label{sec:in-PPAD}
We first show that the equilibrium always exists by the convex compact set version of Brouwer's fixed-point theorem. \\
 {\bf Brouwer's fixed-point theorem}:  
    {\em Every continuous function from a nonempty convex compact subset $\mathcal{K}$ of a Euclidean space to $\mathcal{K}$ itself has a fixed point.}
\begin{theorem}\label{thm:NE_existence}
    Nash equilibrium always exists in every contagion game.
\end{theorem}
\begin{proof}
Let $\mathcal{K}=\{\phi\in \mathbb{R}_+^n| \sum_i\phi_i=1\}$. 
$\mathcal{K}$ is convex and compact. 
Let $U_i(\phi)$ be the individual utility of group $i$ evaluated at point $\phi$. 
We now describe a mapping function from $\phi$ to $\hat{\phi}\in \mathcal{K}$, i.e. $f(\phi) = \hat{\phi}$.
Define $U_{max}(\phi)=\max_i U_i(\phi)$.
Define set $U^-=\{i|U_i(\phi)<U_{max}\}$ and set $U^+=\{i|U_i(\phi)=U_{max}\}$.
Let $0<\alpha<1$ be a small constant. 
For all $i\in U^-$, let $\hat{\phi}_i = \max(\phi_i - \alpha(U_i(\phi) - U_{max}), 0)$, i.e. reduce $\phi_i$ if group $i$'s individual utility is not max.
Let $\Delta = \sum_{i\in U^-} (\phi_i - \hat{\phi}_i)$, the total reduction in $\phi$ for groups in $U^-$.
For all $i\in U^+$, let $\hat{\phi}_i = \frac{\Delta}{|U^+|}$, evenly distribute the reduction into groups in $U^+$.

The intuition is that when the current point is not a fixed point then for all $\phi_i>0$ if $U_i(\phi)<U_{max}(\phi)$, then $\hat{\phi}_i$ is reduced.

The composition of continuous functions is also continuous. The $\max()$ function is continuous, thus $U_{max}$, $\Delta$ are continuous.
Therefore $f(\phi)$ as a composition of continuous functions, is also a continuous function, mapping from $\mathcal{K}$ back to $\mathcal{K}$. By \textbf{Brouwer's fixed-point theorem}, $f$ has a fixed point $\phi^*$, such that $f(\phi^*) = \phi^*$.
At the fixed point $\phi^*$, we have
$
\forall i|\phi^*_i>0 \implies
U_i(\phi^*) = U_{max}(\phi^*)
$.
Therefore it is a Nash equilibrium at the fixed point $\phi^*$. The Nash equilibrium always exists.
\end{proof}

\subsection{PPAD-hardness of the Contagion Game}
We now show that computing the Nash equilibrium in Contagion games with arbitrary interacting policy sets (i.e. with arbitrary $\beta$ matrix ) is PPAD-hard. 
We start with the problem of computing the 2-player Nash equilibrium, termed here as {\bf 2-NASH} which has been shown to be PPAD-complete\cite{chen2009settling}. 
  {\bf 2-NASH} has a polynomial reduction to  {\bf SYMMETRIC NASH}\cite{nisan2007introduction},  which is to find a symmetric Nash equilibrium when the two players have the same strategy sets and their utilities are the same when the player strategies are switched. 
Denote by {\bf CONTAGION NASH} 
the problem of computing a Nash equilibrium in the contagion game with interacting policy sets ({\bf Model A1}).
We reduce 2-player {\bf SYMMETRIC NASH} to {\bf CONTAGION NASH}, showing that it is PPAD-hard.

A 2-player symmetric game has an $n\times n$ payoff matrix $A$ for both players. 
Both players have the same strategy set of size $n$. 
$A_{i,j}$ is the payoff of player 1(2, respectively) when player 1(2, respectively) plays strategy $i$ and the other player plays strategy $j$. 
A Nash equilibrium is symmetric when both players have the same mixed strategy $\sigma^* = [\sigma^*_1, \cdots, \sigma^*_n]^T$. 
We first observe that we can transform any symmetric game with payoff matrix $A$ into a symmetric game with all negative payoffs $\overline{A}$. Let $C=\max_{i,j} A_{i,j}$ (assume for simplicity $C>0$). Define $\overline{A} = A - 2C$.
A Nash equilibrium in $A$ is clearly a Nash equilibrium in $\overline{A}$.
Let $U^*$ be the utility of the symmetric equilibrium $\sigma^*$ in the game defined by $A$. 
Define $\overline{U} = U^* -2C<0$. $\sigma^*$ is also a symmetric equilibrium for the payoff matrix $\overline{A}$ satisfying 
\begin{equation}\label{eq:condition1}
\begin{cases}
    \sum_{j\in NE} \overline{A}_{i,j} \sigma^*_j = \overline{U},\quad& \forall i \in NE\\
    \sum_{j\in NE} \overline{A}_{i,j} \sigma^*_j \leq \overline{U},\quad& \forall i \notin NE
\end{cases}\quad
\text{where }NE=\{i|\sigma^*_i>0\}
\end{equation}

We now discuss the properties of the final size $S_i(\infty)$ at equilibrium $\phi^*$ in the interacting policy sets. 
For simplicity we denote the final size of group $i$ by $S_i$. 
Denote $\overline{S}_i = \frac{S_i}{\phi_i}$. 
For all group $i$, denote 
$X_i=\sum_{j=1}^n \frac{\beta_{i,j}}{\gamma}(S_j-\phi^*_j)
= \sum_{j=1}^n \frac{\beta_{i,j}}{\gamma}\phi^*_j(\overline{S}_j - 1)$.
Applying Equation {\bf$(11)$} from \cite{magal2016final}, 
the final size satisfies
$S_i = S_i(0)\cdot e^{X_i} = (1-\epsilon)\phi^*_i \cdot e^{X_i}$.
And $\overline{S}_i = (1-\epsilon)e^{X_i}$. Since the final size $0\leq S_i<S_i(0)=(1-\epsilon)\phi^*_i$, we have $X_i<0$ and $0\leq \overline{S}_i<(1-\epsilon)$. 
Define set $\widetilde{NE}=\{i|\phi^*_i>0\}$. 
Since for all $i \notin \widetilde{NE}$, $\phi^*_i=0$, we have
$X_i = \sum_{j\in \widetilde{NE}} \frac{\beta_{i,j}}{\gamma}\phi^*_j(\overline{S}_j - 1)$.
Recall that the individual utility $U_i=p_i\cdot h(\overline{S}_i)$. We choose $h$ to be an identity function, i.e. $U_i=p_i\cdot\overline{S}_i$, where $p_i\geq0$.
Suppose the equilibrium individual utility is $N >0 $, we have the following
\begin{equation}\label{eq:condition2}
\begin{cases}
    p_i\cdot \overline{S}_i = N,\quad &\forall i \in \widetilde{NE}\\
    p_i\cdot \overline{S}_i \leq N,\quad &\forall i \notin \widetilde{NE}
\end{cases}\quad
\text{where }\widetilde{NE}=\{i|\phi^*_i>0\}
\end{equation}
Since for all $i\in \widetilde{NE}$, $\overline{S}_i = \frac{N}{p_i}$, we get
$X_i = \sum_{j\in \widetilde{NE}} \frac{\beta_{i,j}}{\gamma}(\frac{N}{p_j} - 1)\phi^*_j$.
Define function
$f(N) = \frac{\ln\frac{N}{1-\epsilon}}{N-1}, 0 < N\leq 1-\epsilon$.
\begin{lemma}\label{lm:inverse_function}
    $f(N)$ is a monotonically decreasing function in the domain $(0, 1-\epsilon]$.
\end{lemma}

Thus for all $\overline{U}<0$, there exists a unique $N$ such that $f(N) = -\overline{U}$, in other words, 
$N=f^{-1}(-\overline{U})$.

We now construct the reduction.
From the {\bf SYMMETRIC NASH} instance $\overline{A}$, we construct an instance, $\overline{C
}$ of {\bf CONTAGION NASH} with
$\gamma=1,
\epsilon=0.0001,
\beta = -\overline{A}$ and $
p_i = 1, \forall i$.
The construction can be done in polynomial time. 
Let $\phi^* = \sigma^*$, and thus the sets $\widetilde{NE} = NE$.
We show the following. 
\begin{lemma}\label{lm:reduction}
    $\sigma^*$ is a Nash equilibrium of $\overline{A}$ $\iff$ $\phi^*$ is a Nash equilibrium of $\overline{C}$.
\end{lemma}
\begin{proof}
    Suppose $\sigma^*$ is a {\bf SYMMETRIC NASH} equilibrium.
    \begin{enumerate}[label=(\roman*)]
        \item $\forall i\in NE$,
        \begin{align*}
            \sum_{j\in NE} \overline{A}_{i,j}\sigma^*_j = \overline{U} \iff
            \sum_{j\in NE} -\beta_{i,j}\phi^*_j = -\frac{\ln\frac{N}{1-\epsilon}}{N-1} \iff
            \sum_{j\in NE} \beta_{i,j}(N-1)\phi^*_j =\ln\frac{N}{1-\epsilon} \iff\\
            \text{(Recall that } \gamma=1\text{ and } p_i=1,\forall i\text{)}\quad
            \sum_{j\in NE} \beta_{i,j}(\frac{N}{p_j}-1)\phi^*_j =\ln\frac{N}{(1-\epsilon)p_i} \iff\\
            X_i=\ln\frac{N}{(1-\epsilon)p_i}\iff 
            p_i(1-\epsilon)e^{X_i}=N\iff
            p_i\overline{S}_i=N
        \end{align*}
        \item $\forall i\notin NE$,
        \begin{align*}
            \sum_{j\in NE} \overline{A}_{i,j}\sigma^*_j \leq \overline{U} \iff
            \sum_{j\in NE} -\beta_{i,j}\phi^*_j \leq -\frac{\ln\frac{N}{1-\epsilon}}{N-1} \iff
            \sum_{j\in NE} \beta_{i,j}\phi^*_j \geq \frac{\ln\frac{N}{1-\epsilon}}{N-1} \iff\\
            \text{(Recall that }N-1<0\text{)}\quad
            \sum_{j\in NE} \beta_{i,j}(N-1)\phi^*_j \leq \ln\frac{N}{1-\epsilon} \iff\\
            \sum_{j\in NE} \beta_{i,j}(\frac{N}{p_j}-1)\phi^*_j \leq\ln\frac{N}{(1-\epsilon)p_i} \iff
            X_i\geq\ln\frac{N}{(1-\epsilon)p_i}\iff 
            p_i(1-\epsilon)e^{X_i}\leq N\iff\\
            p_i\overline{S}_i\leq N
        \end{align*}
    \end{enumerate}
    Since $\widetilde{NE}=NE$,
    $
        \text{$\sigma^*$ is a {\bf SYMMETRIC NASH} equilibrium}\iff
        \text{Condition {\bf (\ref{eq:condition1}})}
        \iff$\\
    $
        \text{Condition {\bf (\ref{eq:condition2})}}
        \iff\text{$\phi^*$ is a {\bf CONTAGION NASH} equilibrium}
    $.\\
    In terms of the numerical error, $\delta$ in {\bf SYMMETRIC NASH} translates to $e^\delta > \delta$, which makes sure that if {\bf CONTAGION NASH} is computed by a $\delta$-approximation, the corresponding {\bf SYMMETRIC NASH} is no worse than a $\delta$-approximation.
\end{proof}

This completes the polynomial reduction from {\bf SYMMETRIC NASH} to {\bf CONTAGION NASH}, proving that  {\bf CONTAGION NASH} is PPAD-hard. 
\begin{theorem}\label{thm:PPAD}
     {\bf CONTAGION NASH} in games with interacting policy sets ( with arbitrary $\beta$ matrix) is PPAD-hard.
\end{theorem}
In the following sections, we focus our attention on the special cases of uniform interaction policy sets and separable policy sets.
\section{Uniform Interaction Policy Sets}
\label{sec:decomp}
Given the hardness of the general policy game, in this section we focus on a special case of the game which is the case with uniform interaction policy sets, i.e. {\bf Model A2}. Recall that the difference from the general interacting model is that we require each entry in the $\beta$ matrix to be
$\beta_{i,j} = \kappa_i \kappa_j \beta_0, \forall i,j$,
where $1=\kappa_1>\kappa_2>\cdots>\kappa_n>0$ and $\gamma/\beta_0 <1$. 
Each group $i$ has a parameter $\kappa_i$ that uniformly determines its interaction with all other groups.

\noindent
{\bf{Preliminaries:}}
For simplicity we denote the  size $S_i(\infty)$ of group $i$ by $S_i$.  Denote by $X_i=\sum_{j=1}^n \frac{\beta_{i,j}}{\gamma}(S_j-\phi_j),  \forall i=1,\cdots, n$. 
$S_i$ satisfies $S_i=S_i(0)\cdot e^{X_i}$ \cite{magal2016final}. 
For better clarity, in later proofs we may denote $e^x$ by $exp(x)$ when $x$ is a long expression.
We assume that the vector $p$ satisfies that $p_1>p_2>\cdots>p_n$. 
\begin{lemma}\label{lm:decreasing_payments}
W.L.O.G. the following property holds for the groups:    $p_1>p_2>\cdots>p_n$.
\end{lemma}
\begin{proof}
Recall that $S_i = (1-\epsilon)\phi_ie^{X_i}$.
Given any $\phi$,
$X_i = \sum_j \frac{\beta_{i,j}}{\gamma}(S_j-\phi_j)
= \sum_j \frac{\kappa_i \kappa_j \beta_0}{\gamma}(S_j-\phi_j)
=\kappa_i X_0$,
where $X_0=\sum_j \frac{\kappa_j \beta_0}{\gamma}(S_j-\phi_j)<0$.
$X_0$ is independent of group $i$. 
The individual utility
$U_i=p_i\cdot h(\overline{S}_i)
=p_i\cdot h((1-\epsilon)\cdot e^{\kappa_i X_0})$.
When $i<j$, since $h$ is increasing, if  $p_i\leq p_j$, we always have $U_i<U_j$.  We may thus remove group $i$ as it will not be in any Nash equilibrium or social optimum. Therefore for any $i<j$, we can assume  $p_i>p_j$.
\end{proof}
We first establish a bound on all the final sizes (proof is in the appendix).
\begin{lemma}\label{lm:n_group_final_upper}
$\sum_{i=1}^n \frac{\kappa_i^2 \beta_0}{\gamma} S_i < 1$.
\end{lemma}

\subsection{Polynomial Time Convex Programming Approach to Compute the Final Size}\label{sec:decomp_convex}
In this section we consider computation of the final size $S_i(\infty)$, and provide a polynomial algorithm to compute an approximation to this size.

Our strategy is to define a convex program that computes the fixed point of functions that defines $S_i(\infty)$.
Given a point $s=[s_1,s_2,\cdots,s_n]^T$, we define function
$f_i(s)=s_i - (1-\epsilon)\phi_i\cdot e^{X_i},\forall i$.
The following convex program finds the final sizes as its unique optimum solution, illustrated by {\bf Figure \ref{fig:convex}}.
\begin{align*}
    \min_s\quad &\sum_{i=1}^n s_i&\\
    \mathrm{s.t.}\quad &f_i(s)\geq 0,\quad  i=1,\cdots, n\\
    &0\leq s_i\leq(1-\epsilon)\phi_i,\quad i=1,\cdots, n
\end{align*}
\begin{figure}
    \centering
    \includegraphics[width = 0.7\textwidth]{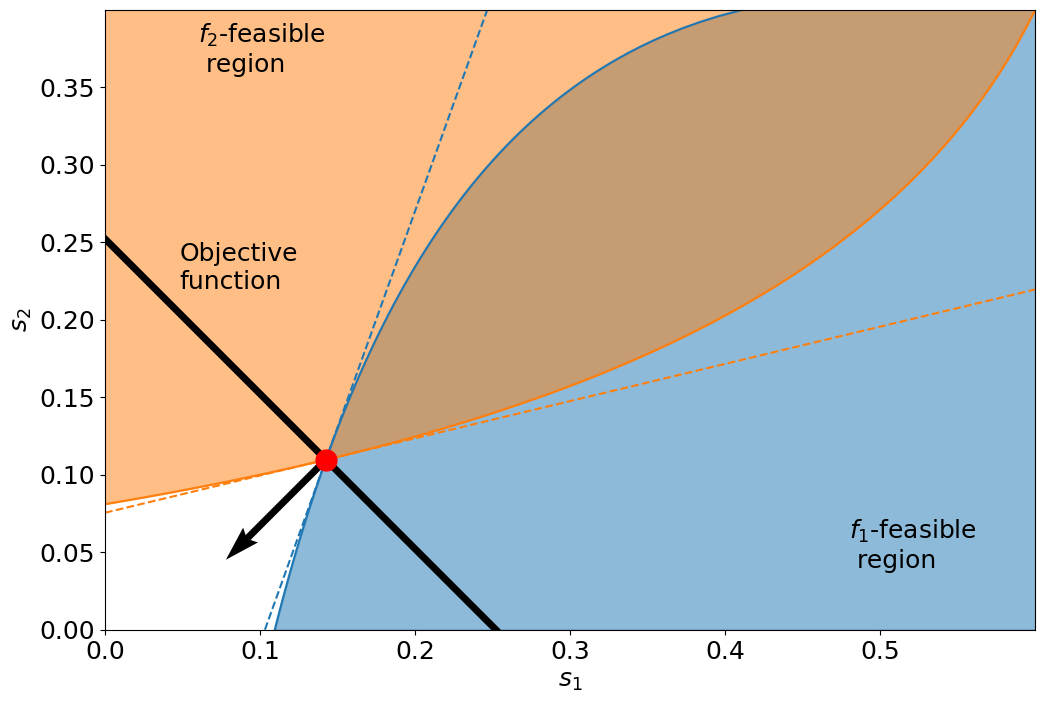}
    \caption{A demonstration of the convex program in a 2-group setting. The red point is the final size point $F^*$.}
    \label{fig:convex}
\end{figure}
We show that this  program is convex by proving that  each $f_i$ is concave and thus together with  $f_i\geq 0$ our domain is  a convex region.
\begin{lemma}\label{lm:f_i_concave}
    $f_i(s)$ is concave.
\end{lemma}
Let $F^*=(S_1, S_2, \cdots, S_n)$ be the final size point.
$F^*$ is feasible to the convex program since $f_i(F^*)=0,\forall i$. 
We next show that $F^*$ is the unique optimum point to this program.
Let $H^*$ be the objective hyperplane passing through $F^*$. 
Any point $s=(s_1,s_2,\cdots,s_n)$ on $H^*$ satisfies
$\sum_{i=1}^n (s_i-S_i)=0$.
We show that any point $s\neq F^*$ on $H^*$ is infeasible. Since the feasible region is convex, it suffices to show that with a small deviation $\Delta\in \mathbb{R}^n$, $s=F^*+\Delta$ is infeasible.
Since point $s$ is on $H^*$, we get $\sum_{i=1}^n \Delta_i=0$.
We partition the components of $\Delta$ into two sets:
$\Delta_- =\{i|\Delta_i<0\}$ and 
$\Delta_+ =\{i|\Delta_i\geq 0\}$.
It is obvious that $\sum_{i \in \Delta_-}\Delta_i = - \sum_{i \in \Delta_+}\Delta_i$.

The Jacobian of $f_i$ at point $F^*$ is $J_{f_i} = 
[\frac{\mathrm{d}f_i}
{\mathrm{d}s_1},
\cdots,
\frac{\mathrm{d}f_i}
{\mathrm{d}s_n}]^T$, where
$
\frac{\mathrm{d}f_i}
{\mathrm{d}s_j} = 
-\frac{\beta_{i,j}}{\gamma} S_i
= -\frac{\kappa_i\kappa_j\beta_0}{\gamma} S_i
,\ \forall j\neq i$ and
$\frac{\mathrm{d}f_i}
{\mathrm{d}s_i} = 
1 - \frac{\beta_{i,i}}{\gamma} S_i 
= 1 - \frac{\kappa_i^2\beta_0}{\gamma} S_i
$.
\begin{lemma}\label{lm:f_i_Delta_infeasible}
    There exists $i$ such that $f_i(F^*+\Delta) <0$.
\end{lemma}

This shows that any point $s=F^*+\Delta$ is infeasible, the final size point $F^*$ is the only feasible point on hyperplane $H^*$. Since the feasible region is convex, it must lie on one side of $H^*$. We observe that vector $\vec{h}=(-1,-1,\cdots,-1)$ is a normal to $H^*$, pointing to the direction that reduces the objective function value. Let point $p_\phi=\Big((1-\epsilon)\phi_1,(1-\epsilon)\phi_2,\cdots,(1-\epsilon)\phi_n\Big)$. It is a feasible point since for all $i$,
$   f_i(p_\phi)
    =(1-\epsilon)\phi_i
    \Big(1 - exp\Big(-\sum_{j=1}^n \frac{\beta_{i,j}}{\gamma}
    \epsilon\phi_j
    \Big)\Big)
    >0$.
For all $i$, $p_\phi[i]>F^*[i]$, thus $\vec{h}\cdot (p_\phi - F^*)<0$, a feasible point $p_\phi$ is on the opposite side of $\vec{h}$, the entire feasible region is on the opposite side of $\vec{h}$. No other feasible point can further reduce the objective function value than $F^*$, therefore $F^*$ is the unique optimum point of the convex program.

There is no bound on the precision of the numbers in the solution; hence we will offer an approximation, based on restricting the location of the solution by solving the convex program using methods like the Ellipsoid method to give the following result:
\begin{theorem}\label{thm:decomp_convex_program}
For the contagion game with uniform interaction policy sets, there exists a convex program to compute $S_i(\infty), \forall i$ in polynomial time.
\end{theorem}


\subsection{Computing the Nash Equilibrium}\label{sec:decomp_NE}
We first look at the individual utility $U_i = p_i\cdot h(\overline{S}_i)$, where $h:R\rightarrow R$ is a concave, monotonically increasing homogeneous function s.t. $h(0)=0$.
We let $h$ be a homogeneous function of degree $d$. Since $h$ is concave, $0\leq d\leq 1$.
We present an algorithm to compute a Nash equilibrium. We first show that at  Nash equilibrium at most  2 groups will participate.
\begin{lemma}\label{lm:2Group_NE}
    For every contagion game with uniform interaction policy sets there exists a Nash equilibrium $\phi^*$ with at most 2 policy groups participating, i.e. $| \{i: \phi^*_i >0 \} | \leq 2$.
\end{lemma}
\begin{proof}
Assume point $\phi^*$ is a Nash equilibrium with $k$ positive components and the corresponding individual utility being $N$. Denote $\overline{S}_i=\frac{S_i}{\phi_i^*}$. 
Let $NE$ be the set of groups in the Nash equilibrium, i.e. $NE=\{i|\phi^*_i>0\}$. 
Since $h$ is homogeneous, for all $i\in NE$,
$
N = p_i \cdot h(\overline{S}_i)\implies
N =h(p_i^{1/d}\cdot\overline{S}_i)
$.
And for all $i \notin NE$, $N \geq p_i\cdot h(\overline{S}_i)$.
Denote $\hN=h^{-1}(N)$ and $\overline{p}_i=p_i^{1/d}$, we have
\[
\begin{cases}
    \forall i \in NE, \quad\phi_i^*>0, \quad \overline{p}_i\cdot\overline{S}_i = \hN,\\
    \forall i \notin NE, \quad\phi_i^*=0, \quad \overline{p}_i\cdot\overline{S}_i \leq \hN
\end{cases}
\]
Note that since $\forall i \notin NE$, $\phi^*_i=0$ and $S_i=0$,
\begin{align*}
X_0 &=\sum_{j=1}^n \frac{\kappa_j \beta_0}{\gamma}(S_j - \phi_j^*) 
= \sum_{j\in NE} \frac{\kappa_j \beta_0}{\gamma}(S_j - \phi_j^*)
= \frac{\beta_0}{\gamma}\sum_{j\in NE} \kappa_j \phi^*_j (\frac{S_j}{\phi^*_j} - 1)
= \frac{\beta_0}{\gamma}\sum_{j\in NE} \kappa_j \phi^*_j (\overline{S}_j - 1)\\
&\text{(replace $\overline{S}_j$ by $\frac{\hN}{\overline{p}_j}$, $\forall j \in NE$)}\quad
= \frac{\beta_0}{\gamma}\sum_{j\in NE} \kappa_j \phi^*_j (\frac{\hN}{\overline{p}_j} - 1)
\end{align*}
The Nash equilibrium $\phi^*$ satisfies the following system of inequalities over the vector-valued variable $\phi$ that defines a polytope over
the space of  non-negative $\phi$:
\[
\begin{cases}
    \frac{\beta_0}{\gamma}\sum_{i\in NE} \kappa_i(\frac{\hN}{\overline{p}_i} - 1)\phi_i = X_0,\\
    \sum_{i\in NE} \phi_i = 1,\\
    \phi_i\geq 0, \quad\forall i \in NE
\end{cases}
\]
Note that $X_0$, which is calculated from $\phi^*$, is a constant independent to the variable $\phi$. The rank of the polytope is at most 2, therefore there exists a basic feasible solution $\hat{\phi}$ with at most 2 non-negative components from the set $NE$. 
If we evaluate the individual utilities at point $\hat{\phi}$, we still get that
$
U_i(\hat{\phi}) = N, \forall i \in NE$ and $U_i(\hat{\phi}) \leq N, \forall i \notin NE
$.
And the at most 2 positive components are from the set $NE$ by construction. Thus we obtain a new point $\phi$ with at most 2 groups and still satisfies the Nash equilibrium conditions. This proves that a Nash equilibrium with at most 2 groups participating exists.
\end{proof}
Now we compute the Nash equilibrium.
First, assume a Nash equilibrium $\phi^*$ with 2 groups exists, namely group $i,j$. 
The 2 groups have the same individual utility, $U_i=U_j=N$.
\[
\begin{aligned}
N=U_i=p_i\cdot h(\hS_i)\implies \ \ \ 
\hN=\hp_i \hS_i=\hp(1-\epsilon)e^{\kappa_i X_0} \implies \ \ 
X_0=\frac{1}{\kappa_i}\ln{\frac{\hN}{(1-\epsilon)\hp_i}}
\end{aligned}
\]
Similarly, we have $X_0=\frac{1}{\kappa_j}\ln{\frac{\hN}{(1-\epsilon)\hp_j}}$.
\[
    \frac{1}{\kappa_i}\ln\frac{\hN}{(1-\epsilon)\overline{p}_i}
    =\frac{1}{\kappa_j}\ln\frac{\hN}{(1-\epsilon)\overline{p}_j}\implies
    \hN = (1-\epsilon)
    \Big(\frac{\overline{p}_j^{\kappa_i}}{\overline{p}_i^{\kappa_j}}
    \Big)^{\frac{1}{\kappa_i - \kappa_j}}
\]
For every $i,j$ pair, compute the Nash equilibrium individual utility $\hN$, then the value of $X_0$. Now solve the following equations to compute $\phi^*_i$ and $\phi^*_j$.
\begin{align*}
&\begin{cases}
    \phi^*_i + \phi^*_j = 1\\
    X_0 = \frac{\beta_0}{\gamma}\Big[\phi^*_i \kappa_i(\frac{\hN}{\overline{p}_i} - 1) 
    + \phi^*_j \kappa_j (\frac{\hN}{\overline{p}_j} - 1)\Big]
\end{cases}\\
\implies
&\begin{cases}
    \phi^*_i=
    \Big[\frac{X_0}{R_0} - \kappa_j(\frac{\hN}{\overline{p}_j}-1)
    \Big] /
    \Big[ \kappa_i(\frac{\hN}{\overline{p}_i}-1) 
    - \kappa_j(\frac{\hN}{\overline{p}_j}-1)\Big]\\
    \phi^*_j=
    \Big[\frac{X_0}{R_0} - \kappa_i(\frac{\hN}{\overline{p}_i}-1)
    \Big] /
    \Big[ \kappa_j(\frac{\hN}{\overline{p}_j}-1) 
    - \kappa_i(\frac{\hN}{\overline{p}_i}-1)\Big]
\end{cases},
\end{align*}
where $R_0 = \beta_0/\gamma$.
Set the remaining group sizes to be 0. For all group $l\neq i,j$, check if its individual utility satisfies $U_l\leq \hN$. If so, we obtain a Nash equilibrium with group pair $i,j$ participating. The entire process can be done in $O(n^3)$.

If no pair produces a Nash equilibrium from the process above, then we try to find equilibrium with only 1 group. Assume group $i$ alone is in the Nash equilibrium, then we have
$\phi^*_i=1$ and 
$\phi^*_j=0,\forall j\neq i$.
We can calculate its final size using the binary search $S_i = FinalSize_i(1)$ from \textbf{Section \ref{sec:sep_LB}}. 
For all other groups $j\neq i$, since $\phi^*_j=0$, $S_j=0$. 
Then we can compute for every group $X_i=\sum_{j=1}^n \frac{\beta_{i,j}}{\gamma}(S_j-\phi_j)$ and the individual utility $\overline{S}_i = (1-\epsilon)e^{X_i}$. 
Let $\hN=\overline{S}_i$. If $\hN\geq \overline{S}_j, \forall j\neq i$, then we obtain a Nash equilibrium with only group $i$ participating. 
The process can be done in $O(n+\log\frac{1}{\delta})$ for 1 group, which will be explained in {\bf Section \ref{sec:sep_LB}}, where $\delta$ is the error bound on $\phi^*$. The overall time complexity is $O(n(n+\log\frac{1}{\delta}))$.

Now we discuss the precision of the algorithm for Nash equilibrium with 1 group. We assume all the inputs are provided in the form of rational number $\pm n_1/n_2$, where $n_1,n_2\leq n_0$. We choose the following value for $\delta$ such that the numerical calculation of Nash equilibrium is correct:

\begin{lemma}\label{lm:decomp_delta_bound}
    $\delta \leq \frac{1}{4 n_0^4}$ guarantees that the group that participates in Nash equilibrium is chosen correctly.
\end{lemma}
Since we have proven the existence of the Nash equilibrium earlier, we are bound to find at least 1 Nash equilibrium from the 2 processes. The overall time complexity is $O(n(n^2+\log \frac{1}{\delta}))$.
\begin{theorem}\label{thm:decomp_Nash_algo}
    The Nash equilibrium $\phi^*$ of contagion game with uniform interaction policy sets can be approximated by $O(n(n^2+\log \frac{1}{\delta}))$ operations where $\delta $ is bounded above by  $ \frac{1}{4 n_0^4}$.
\end{theorem}

\subsection{Price of Anarchy}\label{sec:decomp_POA}
In this section we present results on the price of anarchy(POA) for the contagion game model with uniform interaction policy sets. 

\paragraph{Reducing to Linear utility functions:} For every Nash equilibrium with the corresponding individual utility $N$, let $\phi^{NE} \in R^n$ be the Nash equilibrium solution using the utility function $p_i\cdot h(\hS_i)$. 
For each group $i$, let $N_i$ be its individual utility and $S_i^{NE}$ be the corresponding final size the current at the Nash equilibrium. 
Note that for all group $i$ participating in the Nash equilibrium, $N_i=N$.
Denote $\hS_i^{NE}=\frac{S_i^{NE}}{\phi_i^{NE}}$.
We construct a new utility function $g(\hS_i) = k(\hS_i-\hS_i^{NE}) + N_i
= k\hS_i +(N_i-k\hS_i^{NE})$, where $k=\left.\frac{\mathrm{d}h_i}{\mathrm{d}\hS_i}\right|_{\hS_i^{NE}}$.
The function $g_i$ defines a tangent line at the current Nash equilibrium point $(\hS_i^{NE}, N_i)$, which is an affine utility function.
Since $g_i$ is a tangent to $h_i$, which is a non-negative concave function, we have $g_i(\hS_i) \geq h_i(\hS_i), \forall \hS_i \in R^{+}$.

We show that the POA using  the original utility functions ${\mathcal H} = \{ h_i \}_i$ is bounded by the POA using these new affine utilities.
Denote the POA using function ${\mathcal H}$ and ${\mathcal G} = \{ g_i \}_i$ by $POA_h$, $POA_g$, respectively.
\begin{lemma}\label{lm:POA_fg}
    $POA_g\geq POA_h$.
\end{lemma}

\begin{proof}[Proof of {\bf Lemma \ref{lm:POA_fg}}]
Denote the social welfare at optimum with utility function $h$ and $g$ by $SW_h^{OPT}$ and $SW_g^{OPT}$, respectively. 
Denote the minimum valued social welfare from amongst all Nash equilibria using utility function $h$ and $g$ by $SW_h^{NE}$ and $SW_g^{NE}$, respectively. We have  $POA_h=\frac{SW^{OPT}_h}{SW_h^{NE}}$, $POA_g = \frac{SW^{OPT}_g}{SW_g^{NE}}$.

We first show that $\phi^{NE}$ is still a Nash equilibrium when the utility function is  $g$. We show this as follows: Since $h_i$ is an increasing function, the slope of the tangent at $\hS_i^{NE}$ which is $k$, is $positive$ and $g_i$ is a linear increasing function. 
By definition, $g_i(\hS_i^{NE}) = h_i(\hS_i^{NE})$. Thus, for all group $i$ , $U_i=g_i(\hS_i^{NE})=h_i(\hS_i^{NE}) = N_i$. 
All groups still satisfy that $\phi_i >0 \implies U_i\geq U_j,\forall j$. With utility functions $\{ g_i \}_i$, the current point is still a Nash equilibrium, $SW_g^{NE} = SW_h^{NE} = N$.
        
Next, since $g_i(\hS_i)\geq h_i(\hS_i), \forall \hS_i$, the social welfare using $g_i$ is greater than or equal to that using $h_i$ evaluated at any point, we have $SW^{OPT}_g \geq SW^{OPT}_h$. 

Therefore $POA_g = \frac{SW^{OPT}_g}{SW_g^{NE}} \geq \frac{SW^{OPT}_h}{SW_h^{NE}} = POA_h$.
\end{proof}
Since $g_i, \forall i$ is an affine function, $POA$ using ${\mathcal G}$ is bounded by the $POA$ of utility functions chosen from the affine family,  namely $\hat{g_i}(\hS_i)=a_i \hS_i+b_i$, where $a_i>0,b_i\geq 0$.
With a proof similar to {\bf Lemma~\ref{lm:decreasing_payments}}, we further assume that $a_1> a_2> \cdots > a_n$.
\begin{lemma}\label{lm:POA_affine}
The price of anarchy,    $POA_{\hat{g}}$ is maximized when $b_i=0, \forall i$.
\end{lemma}
\begin{proof}[Proof of {\bf Lemma \ref{lm:POA_affine}}]
The social welfare is $\sum_i \phi_i \hat{g}_i(\hS_i) = 
\sum_i \phi_i (a_i\hS_i + b_i) = \sum_i \phi_ia_i \hS_i + \sum_i \phi_i b_i$. 
Denote $\hat{b}=\sum_i \phi_i b_i$.
By contradiction, assume $\hat{b}>0$. Let $POA_{\hat{g}}=\frac{SW^{OPT}}{SW^{NE}}$ for the function $\hat{g}$ (we omit the subscript $\hat{g}$ for simplicity for the rest of the proof), 
where $SW^{OPT}= \overline{SW}^{OPT}+ \hat{b}$ and $SW^{NE}= \overline{SW}^{NE}+ \hat{b}$. 
Since $\overline{SW}^{OPT}>\overline{SW}^{NE}$ and $\hat{b}>0$,
\[
\begin{split}
   &\overline{SW}^{OPT}\cdot \hat{b} > \overline{SW}^{NE}\cdot \hat{b}\\
   \implies
   &\overline{SW}^{OPT}\cdot \overline{SW}^{NE} +
   \overline{SW}^{OPT}\cdot \hat{b} > 
   \overline{SW}^{OPT}\cdot \overline{SW}^{NE} +
   \overline{SW}^{NE}\cdot \hat{b}\\
   \implies
   &\overline{SW}^{OPT}(\overline{SW}^{NE} + \hat{b}) >
   \overline{SW}^{NE}(\overline{SW}^{OPT} + \hat{b})\\
   \implies
   &\frac{\overline{SW}^{OPT}}{\overline{SW}^{NE}} >
   \frac{\overline{SW}^{OPT} + \hat{b}}{\overline{SW}^{NE} + \hat{b}} = POA
\end{split}
\]
which is a contradiction. Thus $POA$ is achieved when $b_i=0,\forall i$ for the class of affine functions $\hat{g}$.
\end{proof}
We can focus on the affine utility function $\hat{g}_i(\hS_i) = a_i \hS_i$. 
We first show the lower bound of the group 1's individual utility $U_1$.
Recall that $\beta_{i,j}=\kappa_i \kappa_j \beta_0$ with $1=\kappa_1>\kappa_2>\cdots>\kappa_n>0$, so group 1 has the highest $\beta$'s. We show that $U_1$ is lowest when $\phi_1=1$. 
Let $\phi_{END}=[1,0,\cdots,0]^T$.
To show $U_1(\phi_{END}) \leq U_1(\overline{\phi}),
\forall \overline{\phi}=[\overline{\phi}_1, \overline{\phi}_2\, \cdots, \overline{\phi}_n]^T$, we show the following.

\begin{lemma}\label{lm:S1_bar_lower_bound}
$\overline{S}_1(\phi_{END}) \leq \overline{S}_1(\overline{\phi}),
\forall \overline{\phi}$.
\end{lemma}

Thus $U_1(\phi_{END})\leq U_1(\overline{\phi})$, for all $\overline{\phi}$.
Group 1's individual utility is lowest when its group size is 1.
For any Nash equilibrium point $\phi^{NE}$ with the corresponding individual utility $N$, if group 1 is participating, $N=U_1(\phi^{NE}) \geq U_1(\phi_{END})$; if group 1 is not participating, $N\geq U_1(\phi^{NE}) \geq U_1(\phi_{END})$.
Therefore $U_1(\phi_{END})$ is a lower bound of individual utility of any Nash equilibrium.

When at $\phi_{END}=[1,0,0,\cdots ,0]^T$, the entire population is in group 1, there is no interaction with other groups. 
We may apply the lower bound $LB$ of $\overline{S}_1$ with $\phi_1=1$ from the separable model obtained in {\bf Section \ref{sec:sep_LB}}.
Since $\phi_1=1$, $\overline{S}_1=S_1$.
Let the individual utility at Nash equilibrium be $N>0$, we get
$N \geq  U_1(\phi_{END})=a_1\overline{S}_1(\phi_{END})
    \geq \frac{a_1}
    {B_1}$,
where $B_1 = e^{\frac{\beta_0}{\gamma}} / (1 - \epsilon) - \frac{\beta_0}{\gamma}$.
The social optimum is
\[OPT = \max_{\phi} \sum_i UG_i 
=\max_\phi \sum_i \phi_i a_i \hS_i
= \max_\phi \sum_i a_i S_i \]
Since for all $i$, $S_i\leq S_i(0) = (1 - \epsilon) \phi_i$,
\begin{align*}
OPT &\leq \max_\phi \sum_i a_i(1-\epsilon)\phi_i
\leq (1-\epsilon)a_1\sum_i\phi_i = (1-\epsilon)a_1
\end{align*}
We next establish an upper bound on the optimum and subsequently determine the PoA.

The optimum can be determined by
\begin{align}\label{Opt_exact}
\max     \sum_i {a_i}S_i\\
    S_i = \Phi_i e^{X_0}\\
    \sum_i \Phi_i = 1-\epsilon  
\end{align}
where $\epsilon$ fraction of the population is distributed over the infectious population.
We upper bound the value of the optimum by relaxing the equality conditions and using the fact that $S_i \leq \Phi_i$
\begin{align}\label{Opt_upperBound}
\max     \sum_i {a_i}S_i\\
  \sum_i S_i R_i \leq 1\\
    \sum_i S_i \leq 1, \ 
    S_i \geq 0
\end{align}
Since this is a linear program with two constraints, there exists an optimum solution with two non-zero variables, say $S_j$ and $S_k$. Thus the optimum value is upper bounded by $a_j S_j/R_j+a_kS_k/R_k$. W.l.o.g. assume that $j<k$.

We next establish a lower bound for the Nash solution. Let $z= \sum_i \frac{\kappa_i \beta_0}{\gamma}\phi_i$. The following lemma provides the bound:
\begin{lemma}
Let $ B_i =  (e^z - (1-\epsilon)z )^{\kappa_i}$ where  $z= \sum_i \frac{\kappa_i \beta_0}{\gamma}\phi_i$. Then 
$$    \overline{S}_i \geq \frac{1-\epsilon}{B_i} $$ and
$$ N \geq  a_i \overline{S}_i , \ \forall i$$
\end{lemma}
\begin{proof}
We show a lower bound on $g_i(\overline{S}_i) = (1-\epsilon) e^{\kappa_i X_0}$ where $X_0 = \sum_i \frac{\kappa_i \beta_0}{\gamma}\phi_i (\overline{S}_i -1)$.
Substituting for $\overline{S}_i$ in terms of $\overline {S}_1$ using the fact that 
$\overline{S}_i =(1-\epsilon)^{\kappa_i}\overline{S}_1$, we get $g_i(\overline{S}_i) = (1-\epsilon)e^{\kappa_i X}$ where $X= \sum_i \frac{\kappa_i \beta_0}{\gamma}\phi_i ((1-\epsilon)^{1-\kappa_i}{\overline{S}_1}^{\kappa_i} -1) $
Noting that $(1-\epsilon)^{\kappa_i} \geq \overline{S}_1$ (since $\overline{S}_1 \leq (1-\epsilon)$),
we get $g_i(\overline{S}_i) \geq (1-\epsilon)e^{\kappa_iX_1}$ where $X_1 = \frac{\kappa_i \beta_0}{\gamma}\phi_i (\overline{S}_1 -1)$.
Consider the function $g'_1(\overline{S}_1) = (1-\epsilon)e^{X_1}$ and the region $y \geq g'_1(\overline{S}_1)$. This is a convex region and we need to determine its intersection with the ray defined by $y=\overline{S}_1$ that starts at the origin and intersects $g'(\overline{S}_1)$ to determine  the fixed point $p^*$ of $ g'(\overline{S}_1)$ and use this to determine a lower bound on the value of $\overline{S}_1$ that is a solution to $\overline{S}_1 = g_i(\overline{S}_1)$. We provide a lower bound on $p^*$ by considering the intersection of the tanget to the curve $y = g'(\overline{S}_1) $ at the point $(0,g'(0))$. This tangent is described by the equation:
$$ y = \frac{(1-\epsilon)z}{e^z} \overline{S}_1 + \frac{1-\epsilon}{e^z}$$
The tangent intersects the line $y=\overline{S}_1$ at the point $p_1$ where $\overline{S}_1= \frac{(1-\epsilon)z}{e^z - (1-\epsilon) z} $ This provides us with a lower bound to $\overline{S}_1$, i.e.
$$ \overline{S}_1 \geq \frac{(1-\epsilon)z}{e^z - (1-\epsilon) z} $$
Next $\forall i \neq 1, \ \overline{S}_i = (1-\epsilon)^{1-\kappa_i}{\overline{S}_1}^{\kappa_i} \geq (1-\epsilon)^{1-\kappa_i} \left [ \frac{(1-\epsilon)z}{e^z - (1-\epsilon) z} \right ]^{\kappa_i} $ giving us the result.

\end{proof}

Next assume $N$, the solution to the Nash equilibrium, to be fixed, we set up the following maximization program for the PoA.

\begin{lemma}
\label{lem:UpperBoundOpt}
The optimum solution to the linear program 
\begin{equation*}
\begin{aligned}
\max \quad & \sum a_i S_i \\
\text{s.t.} \quad & \sum S_i \rho_i \leq 1 \\
& \sum S_i \leq 1, \quad S_i \geq 0
\end{aligned}
\end{equation*}
is upper bounded by $\max_i \{  \min \{ a_i/\rho_i, a_i \} \}$:

\end{lemma}
\begin{proof}
This linear program has a solution with at most two non-zero variables. Let them be $S_i$ and $S_j$. Consider the lagrangian:

\begin{equation*}
L = \sum_i a_i S_i + \lambda_1 \left(1 - \sum_i S_i \rho_i\right) + \lambda_2 \left(1 - \sum S_i\right)
\end{equation*}

\begin{equation*}
S_i >0 \implies \frac{\partial L}{\partial S_i} = a_i - \lambda_1 \rho_i - \lambda_2 = 0
\end{equation*}

\begin{equation*}
a_i - \lambda_1 \rho_i = a_j - \lambda_1 \rho_j
\end{equation*}

Note that both inequalities cannot be strict inequality as otherwise one of $S_i$ or $S_j$ can be increased violating optimality of the current solution.

\paragraph{Case 1:} The inequality $\sum_i S_i < 1$ which impies that $\lambda_2 = 0$
\begin{equation*}
\begin{aligned}
\lambda_2 = 0 \implies & a_i - \lambda_1 \rho_i = 0 \implies \lambda_1 = \frac{a_i}{\rho_i} = \frac{a_j}{\rho_j} \\
L &=  a_i S_i + a_jS_j + \left( \frac{a_i}{\rho_i} - a_i S_i -  a_jS_j \right) \\
&= \frac{a_i}{\rho_i}
\end{aligned}
\end{equation*}

\paragraph{Case 2:} The inequality $\sum_i S_i \rho_i < 1$, i.e. $\lambda_1 = 0$.
\begin{equation*}
\begin{aligned}
\lambda_1 = 0 \implies &  
 a_i - \lambda_2 = 0;\\ \mathrm{and} \quad & a_j - \lambda_2 = 0
\end{aligned}
\end{equation*}

So if $a_i \neq a_j$, which is assumed, one of $S_i$ or $S_j$ must be zero. 

Let $S_j=0$ and $\lambda_2 = a_i$. 
Since $S_i =1$ (due to the second inequality being tight and since $S_j=0$), $\rho_i > 1 \implies S_i\rho_i \geq 1 $ and contradicts the inequality being strict. This rules out this case.



\paragraph{Case 3:} Both inequalities are tight.
\begin{align*}
\sum s_i \rho_i &= 1 \\
\sum s_i &= 1 \\
\text{and since we consider two non-zero variables:}\\
s_i \rho_i + s_j \rho_j &= 1 \\
s_i + s_j &= 1 \\
s_i \rho_i + (1 - s_i) \rho_j &= 1 \\
s_i(\rho_i - \rho_j) &= 1 - \rho_j \\
\text{Assume w.l.o.g that } \rho_i > \rho_j; \ \text{And thus } s_i > 0 &\Rightarrow \rho_j < 1 \\
s_i = \frac{1 - \rho_j}{\rho_i - \rho_j} &, \quad s_j = \frac{\rho_i - 1}{\rho_i - \rho_j} \\
s_j > 0 &\Rightarrow \rho_i > 1 \\
\text{Substituting into the objective function which can be rewritten as:}\\ 
\text{Obj }  &= a_i \frac{(1 - \rho_j)}{\rho_i - \rho_j} + a_j \frac{(\rho_i - 1)}{\rho_i - \rho_j} \\
&= \frac{1}{\rho_i - \rho_j} (a_i - a_i \rho_j + a_j \rho_i - a_j) \\
\text{We show that if } \quad \frac{a_j}{\rho_j} &\le \frac{a_i}{\rho_i} 
\text{then Obj }  \le \frac{a_i}{\rho_i}\\
\text{this can be verified as } \quad a_i \rho_i - a_i \rho_i \rho_j + a_j \rho_i^2 - a_j \rho_i &\le a_i \rho_i - a_i \rho_j \\
\equiv a_j (\rho_i^2 - \rho_i) &\le a_i \rho_j (\rho_i - 1) \\
\equiv a_j \rho_i &\le a_i \rho_j \\
\text{or, alternatively if } \frac{a_i}{\rho_i} &\le \frac{a_j}{\rho_j} \quad \text{then Obj }\le \frac{a_j}{\rho_j}\\
\text{Again, this can be verified as }  \quad a_i \rho_j - a_i \rho_i \rho_j + a_j \rho_i \rho_j - a_j \rho_j &\le a_j \rho_i - a_j \rho_j \\
\equiv a_i \rho_j (1 - \rho_i) &\le a_j \rho_i (1 - \rho_j) \\
\equiv a_i \rho_j &\le a_j \rho_i 
\end{align*}
    
\end{proof}









\begin{theorem}\label{thm:decomp_POA}
The price of anarchy for the contagion game with uniform interaction policy sets is bounded   as follows:
\begin{equation*}
\frac{OPT}{N} 
 \leq 1.145 \frac{e^{R_0}}{R_0}
\end{equation*}
where $R_0$ is the maximum reproduction number of the contagion over all policy sets.
\end{theorem}

\begin{proof}
From $\sum_i S_i \rho_i \le 1$ then for all $j$ we have
\[
S_j \rho_j \le 1
\qquad\Rightarrow\qquad
S_j \le \frac{1}{\rho_j}
= \frac{1}{R_0 \kappa_j^2}.
\]

We analyze cases based on 
$\min\{1,1/\rho_j\}$.



\begin{enumerate}

\item 
Let $\rho_j \geq 1$. Then,
from the value of $N$  and $OPT$ (lemma~\ref{lem:UpperBoundOpt}) we get:
\[
\frac{OPT}{N}
\le
\frac{ B_j}{(1-\epsilon)R_0 \kappa_j^2}\]
 where  $B_j \le e^{(\sum_i \kappa_i R_0 \Phi_i) \kappa_j } \leq  e^{ R_0 \kappa_j } $, since $\sum_i \kappa_i R_0 \Phi_i \leq R_0 \sum_i \Phi_i = R_0$.
Define
$
f(\kappa_j)
=
\frac{e^{R_0 \kappa_j}}{\kappa_j^2}$. Differentiating $f(\kappa_j)$ w.r.t. $\kappa_j$,
\[
f'(\kappa_j)
=
e^{R_0 \kappa_j}
\frac{R_0 \kappa_j - 2}{\kappa_j^3}.
\]

\paragraph{Case 1:}  For $\kappa_j R_0\ge 2$, the function is increasing, the maximum over the admissible region occurs at $\kappa_j = 1$.
Thus,
\[
\frac{OPT}{N}
\le
\frac{e^{R_0}}{R_0(1-\epsilon)}.
\]
\paragraph{Case 2:} In this case, $ \kappa_j R_0 < 2$. 


Consider the case when  $\kappa_j^2R_0 \geq 1$. In this case $\kappa_j \geq 1/2$ since $\kappa_j^2R_0 \geq 1  \implies \kappa_j > (1/\kappa_jR_0) \geq 1/2$.

The function $f(\kappa_j)$ is decreasing when $\kappa_jR_0 <2$ and achieves its maximum value when $\kappa_j = 1/2$.
Then  $B_j/(\kappa_j^2R_0) \leq  e^{\kappa_jR_0}  \leq e^{R_0/2}$. When $R_0 \geq 1$, $e^{R_0/2} \leq e^{R_0}/R_0$. 

Next consider when $k_j^2R_0 <1$. Then $\rho_j = \kappa_j^2 R_0 <1$ and we need to consider this case separately.

\item When $\rho_j <1$,  $OPT \leq a_j$ (since $S_j \leq 1$ from Lemma~\ref{lem:UpperBoundOpt}) and 
\[\frac{OPT}{N} \leq \frac{e^{\kappa_jR_0}}{1-\epsilon} \leq 1.145\frac{e^{R_0}}{R_0(1-\epsilon)}
\]
To show the above we need to show that 
\[ \kappa_j R_0 \leq \ln 1.14 + R_0 - \ln R_0\]
or
\[ \sqrt{R_0} \leq \ln (1.145)  +R_0 - \ln R_0\]
since  $\kappa^2_j R_0 \leq  1$. The minimum of $R_0 -\sqrt{R_0} -\ln R_0 + \ln (1.145)$ is provided by the value of $R_0$ satisfying 
$1-1/(2 \sqrt{R_0})-1/R_0=0$, which can be verified to be non-negative.
\end{enumerate}

\end{proof}

\section{Policy Sets over an Interacting Network}
\label{sec:network_policy}

In this section we consider the network interacting model.
In order to analyze this model we need some notations.
Denote $X_i^v=\sum_{(u,j)} \frac{\beta_{i,j}^{v,u}}{\gamma}(S_j^u - \phi_j^u)$. As in the case of the interaction model, {\bf Model A2}, the analysis in \cite{magal2016final} indicates that the final size satisfies $S_i^v = S_i^v(0)\cdot e^{X_i^v} = (1-\epsilon)\phi_i^v \cdot e^{X_i^v}$. 
Denote $\overline{S}_i^v =S_i^v/\phi_i^v = (1-\epsilon) e^{X_i^v}$.

In order to consider the computation of Nash equilibrium, we recall the details of the game theoretic model. 
The individual utility of group $i$ at node $v$ is  $U_i^v = p_i^v\cdot h_i^v(\hS_i^v)$. 
The group utility of group $i$ at node $v$ is $UG_i^v = \phi_i^v U_i^v$. 
Again, the individual utility can be evaluated even when $\phi_i^v=0$.
A Nash equilibrium $\mathcal{N}$ with $\phi^{\mathcal{N}} = [\cdots, \phi_1^{v{\mathcal{N}}},\cdots,\phi_n^{v{\mathcal{N}}},\cdots]^T, \forall v$, is achieved when for all node $v$, $\forall i| \phi_i^{v{\mathcal{N}}} > 0 \implies U_i^v \geq U_j^v, \forall j$
, i.e. the population in node $v$ is only in the group(s) with the highest individual utility.
In this version of game, each group $(v,i)$ only competes with other groups within the same node $v$. We may apply the same mapping function from {\bf Theorem \ref{thm:NE_existence}} separately within every node. This gives an overall mapping function satisfying {\bf Brouwer's fixed-point theorem}, therefore  Nash equilibrium exists. 

For the case of {\bf Model B1}, we note that 
{\bf Section \ref{SIR_model1}} is a special case of the network version where there is only 1 node in the network. Therefore we have a direct reduction from {\bf Model 1}, showing that this case is PPAD-hard.  
While the complexity of {\bf Model B2}  is left as unknown, we next consider the last model {\bf Model B3}. 

    




\subsection{Polynomial Time Convex Programming Approach to Compute the Final Size in the   Network Interaction Model}\label{sec:net_convex}
Given a joint strategy $\phi$ over the network, we first present a similar convex program to compute the final sizes. Given a point $s = [\cdots,s_i^v,\cdots]^T$, define function $f_i^v(s) = s_i^v - (1-\epsilon)\phi_i^v \cdot e^{X_i^v}, \forall (v,i)$. The following convex program finds the final sizes as its unique optimum solution.
\begin{align*}
    \min_s\quad &\sum_{(v, i)} s_i^v&\\
    \mathrm{s.t.}\quad &f_i^v(s)\geq 0,\quad  \forall (v,i)\\
    &0\leq s_i^v\leq(1-\epsilon)\phi_i^v,\quad \forall (v,i)
\end{align*}
Note that the interaction between any pair of groups $(v,i)$ and $(u,j)$, $\beta_{i,j}^{v,u} = \overline{\kappa}_i^v \cdot \overline{\kappa}_j^u \cdot \beta_0$.
We compare the function $f_i^v$ with the function $f_i$ in {\bf Section \ref{sec:decomp_convex}} as listed below:
\begin{align*}
    f_i^v(s) &= s_i^v - (1-\epsilon)\phi_i^v\cdot e^{X_i^v},\quad \forall (v,i),\quad
    \text{where } X_i^v = \sum_{(u,j)} \frac{\beta_{i,j}^{v,u}}{\gamma}(s_j^u - \phi_j^u)
    =\overline{\kappa}_i^v \sum_{(u,j)} \frac{\overline{\kappa}_j^u \beta_0}{\gamma} (s_j^u - \phi_j^u)\\
    f_i(s) &= s_i - (1-\epsilon)\phi_i\cdot e^{X_i}, \quad \forall i,\quad
    \text{where } X_i = \sum_j \frac{\beta_{i,j}}{\gamma}(s_j - \phi_j)
    =\kappa_i \sum_j \frac{\kappa_j \beta_0}{\gamma} (s_j - \phi_j)
\end{align*}
If we map every group $(v,i)$ into $i$ and substitute $\overline{\kappa}$ into $\kappa$, $f_i^v$ becomes equivalent to $f_i$, therefore $f_i^v$ is concave,
the proof of the correctness of the convex program in {\bf Theorem \ref{thm:decomp_convex_program}} also applies. 

\subsection{Computing the Nash Equilibrium in the Uniform  Network Interaction Model}
In this subsection we discuss algorithms to compute the Nash equilibrium. We first prove that at Nash equilibrium either there is  a node with two policy groups participating in the equilibrium or  all nodes of the network have only one policy group participating in the Nash equilibrium. 
We assume that in every node $v$, $p_1^v> p_1^v> \cdots p_n^v$, the proof being similar to {\bf Lemma \ref{lm:decreasing_payments}}.
Denote $X_0 = \sum_{(u,j)}  \frac{\overline{\kappa}_j^u \beta_0}{\gamma} (S_j^u - \phi_j^u)$, $X_0<0$. 
For all group $(v,i)$, 
$X_i^v = \sum_{(u,j)} \frac{\beta_{i,j}^{v,u}}{\gamma}(S_j^u - \phi_j^u) 
= \overline{\kappa}_i^v X_0$.
In a Nash equilibrium $\phi^*$, at node $v$, denote by $N_v$ node $v$'s highest individual utility, and $\overline{S}_i^v =\frac{S_i^v}{\phi_i^v}$. 
Let $NE_v = \{ i | \phi_i^{v*} >0 \}$ be the set of policy groups participating in the equilibrium in node $v$. 
Denote $\hp_i^v=(p_i^v)^{1/d^v}$.
For all $i<j$, $p_i^v>p_j^v\implies\hp_i^v>\hp_j^v$.
For all $i \in NE_v$,
$
N_v = p_i^v\cdot h_i^v(\overline{S}_i^v)
\implies N_v = h^v(\hp_i^v\hS_i^v)$.
For all $i \notin NE_v$, $N_v \geq p_i^v \cdot h_i^v(\overline{S}_i^v)$.
Denote $\hN_v=(h^v)^{-1}(N_v)$, we have
\[
\begin{cases}
\forall i \in NE_v,\quad {\phi_i^v}^*>0,\quad \hp_i^v\cdot \hS_i^v=\hN_v,\\
\forall i\notin NE_v,\quad {\phi_i^v}^*=0,\quad \hp_i^v\cdot \hS_i^v\leq\hN_v
\end{cases}
\]
Since for all $i \notin NE_v$, ${\phi_i^v}^*=0$ and $S_i^v=0$,
\begin{align*}
X_0 &=\sum_{(u,j)}^n \frac{\overline{\kappa}_j^u \beta_0}{\gamma}(S_j^u - \phi_j^{u*}) 
= \frac{\beta_0}{\gamma}\sum_u \sum_{j\in NE_u} \overline{\kappa}_j^u   \phi_j^{u*} (\overline{S}_j^u - 1)\\
&\text{(Replace $\overline{S}_j^u$ by $\frac{\hN_u}{\hp_j^u}$, $\forall j \in NE_u$)}\quad
= \frac{\beta_0}{\gamma}\sum_{j\in NE} \overline{\kappa}_j^u\phi_j^{u*} (\frac{\hN_u}{\hp_j^u} - 1)
\end{align*}
For all $i\in NE_v$,
$\frac{N_v}{p_i^v} = (1-\epsilon)exp\Big(\frac{\overline{\kappa}_i^v \beta_0}{\gamma} \sum_u \sum_{j \in NE_u} \overline{\kappa}_j^u \phi_j^{u*} (\frac{N_u}{p_j^u} - 1)\Big)\implies $
\begin{align*}
    \frac{\gamma}{\overline{\kappa}_i^v \beta_0} \ln\frac{N_v}{(1-\epsilon)p_i^v}=
    \sum_u \sum_{j\in NE_u} \overline{\kappa}_j^u (\frac{N_u}{p_j^u} - 1)\phi_j^{u*}
\end{align*}
The Nash equilibrium $\phi^*$ satisfies the following system of inequalities over the vector-valued variable $\phi$ that defines a polytope over the space of non-negative $\phi$:
\[
\begin{cases}
    \sum_v\sum_{i\in NE_v} \overline{\kappa}_i^v(\frac{\hN_v}{\hp_i^v} - 1)\phi_i^{v} = X_0\\
    \sum_{i\in NE_v} \phi_i^{v} = 1, \quad \forall v\\
    \phi_i^{v}\geq 0, \quad\forall (i,v)| i\in NE_v
\end{cases}
\]
Note that $X_0$, which is calculated from $\phi^*$, is a constant independent to the variable $\phi$.
The rank of this polytope is $m+1$, but because of $\sum_{i\in NE_v} \phi_i^{v} = 1,\forall v$, the rank is at least $m$, there is at least one positive component in $\phi^v$, for every node $v$. 
Give any Nash equilibrium $\phi^*$, we can construct a new equilibrium $\phi$ satisfying one of the two following cases.
\begin{enumerate}[label=(\roman*)]
\item 
\label{NetworkCase1}
Only one node $v$ has two groups $i,j$ participating in the Nash equilibrium, with $U_i^v = U_j^v \geq U_l^v,\forall l$ and $\phi_i^{v}, \phi_j^{v}>0$, $\phi_l^{v}=0, \forall l\neq i,j$. The rest of the nodes all have only 1 dominating group, namely $i$, with $U_i^v \geq U_j^v, \forall j$ and $\phi_i^{v}=1$, $\phi_j^{v}=0, \forall j\neq i$.
\item 
\label{NetworkCase2}
Every node $v$ has only one group $i$ dominating all other groups.
\end{enumerate}
We proceed to present algorithms to compute the Nash equilibrium in both cases.\\ \\
\noindent
{\bf{Case \ref{NetworkCase1}}:}
In this case we determine the node  that has two groups participating in the Nash equilibrium. To do so, we iterate over every possible candidate combination of node $v$ and group $i,j$ in $v$ (this is a polynomial number of combinations). Assume node $v$ is the node with 2 groups in the equilibrium, namely $i,j$. 
\begin{align*}
    &U_i^v = U_j^v\implies
    \hp_i^v \overline{S}_i^v = \hp_j^v \overline{S}_j^v \implies
    \hp_i^v e^{X_i^v} = \hp_j^v e^{X_j^v}\\
    &\implies
    \ln \frac{\hp_i^v}{\hp_j^v} = X_j^v - X_i^v = (\overline{\kappa}_j^v - \overline{\kappa}_i^v)X_0
    \implies X_0 = \frac{1}{\overline{\kappa}_j^v - \overline{\kappa}_i^v}\ln \frac{\hp_i^v}{\hp_j^v}
\end{align*}
We denote  $X_{i,j}^v = \frac{1}{\overline{\kappa}_j^v - \overline{\kappa}_i^v}\ln \frac{\hp_i^v}{\hp_j^v}$.
Note that $X_0$ and hence $X_{i,j}^v$ is always negative.
Now we can compute every group's individual utility, across all nodes.
$
U_l^v = p_l^v \cdot h^v((1-\epsilon) e^{\overline{\kappa}_l^v X_{i,j}^v}$), $\forall (v,l)$.
We first check if $U_i^v$ and $U_j^v$ indeed dominate all other groups in node $v$ by comparing them to all other groups' individual utility $U_l^v,\forall l\neq i,j$.
If not, we move to the next candidate $(v,i,j)$. For each node $u\neq v$, we find the group $l$ with the highest individual utility and set $\phi_l^u=1$. We solve for $\phi_i^v, \phi_j^v$ from the following equations.
\begin{equation}
\label{eq:phi_ij}
\begin{cases}
    \sum_v\sum_{i\in NE_v} \overline{\kappa}_i^v(\frac{\hN_v}{\hp_i^v} - 1)\phi_i^{v} = X_{i,j}^v\\
    \phi_i^v + \phi_j^v = 1
\end{cases}   
\end{equation}
If both $\phi_i^v, \phi_j^v$ are non-negative, we have found a Nash equilibrium. Otherwise we move onto the next candidate $(v,i,j)$. The number of candidates is $O(mn^2)$, for each candidate we spend $O(mn)$ steps. The time complexity for {\bf Case \ref{NetworkCase1}} is $O(m^2n^3)$.
We summarize the algorithm in {\bf Algorithm \ref{alg:net1}} in {\bf Appendix \ref{appendix:algo}}.

\noindent
{\bf{Case \ref{NetworkCase2}}:}
Recall the value $X_{i,j}^v = \frac{1}{\overline{\kappa}_j^v - \overline{\kappa}_i^v}\ln \frac{\hp_i^v}{\hp_j^v}$. When $X_0= X_{i,j}^v$, we have the property  that in node $v$ the individual utility of group $i$ and $j$ are equal. Without loss of generality, assume $i<j$, thus $\hp_i^v > \hp_j^v$ and $\overline{\kappa}_i^v > \overline{\kappa}_j^v$. We show the following.
\begin{lemma}\label{lm:indi_utility_range}
    When $X_0 < X_{i,j}^v$, $U_i^v < U_j^v$, and vice versa.
\end{lemma}
\begin{proof}
    \begin{align*}
    &X_0 < X_{i,j}^v \implies
    (\overline{\kappa}_i^v - \overline{\kappa}_j^v)X_0 
    < (\overline{\kappa}_i^v - \overline{\kappa}_j^v) X_{i,j}^v \implies\\
    &exp\Big((\overline{\kappa}_i^v - \overline{\kappa}_j^v)X_0\Big) 
    <
    exp\Big((\overline{\kappa}_i^v - \overline{\kappa}_j^v) X_{i,j}^v \Big)
    = \frac{\hp_j^v}{\hp_i^v} \implies
    \frac{\overline{S}_i^v}{\overline{S}_j^v} < \frac{p_j^v}{p_i^v} \implies
    U_i^v < U_j^v
\end{align*}
Similarly, when $X_0 > X_{i,j}^v$, $U_i^v > U_j^v$.
\end{proof}

\begin{figure}
    \centering
    \includegraphics[width = 0.6\textwidth]{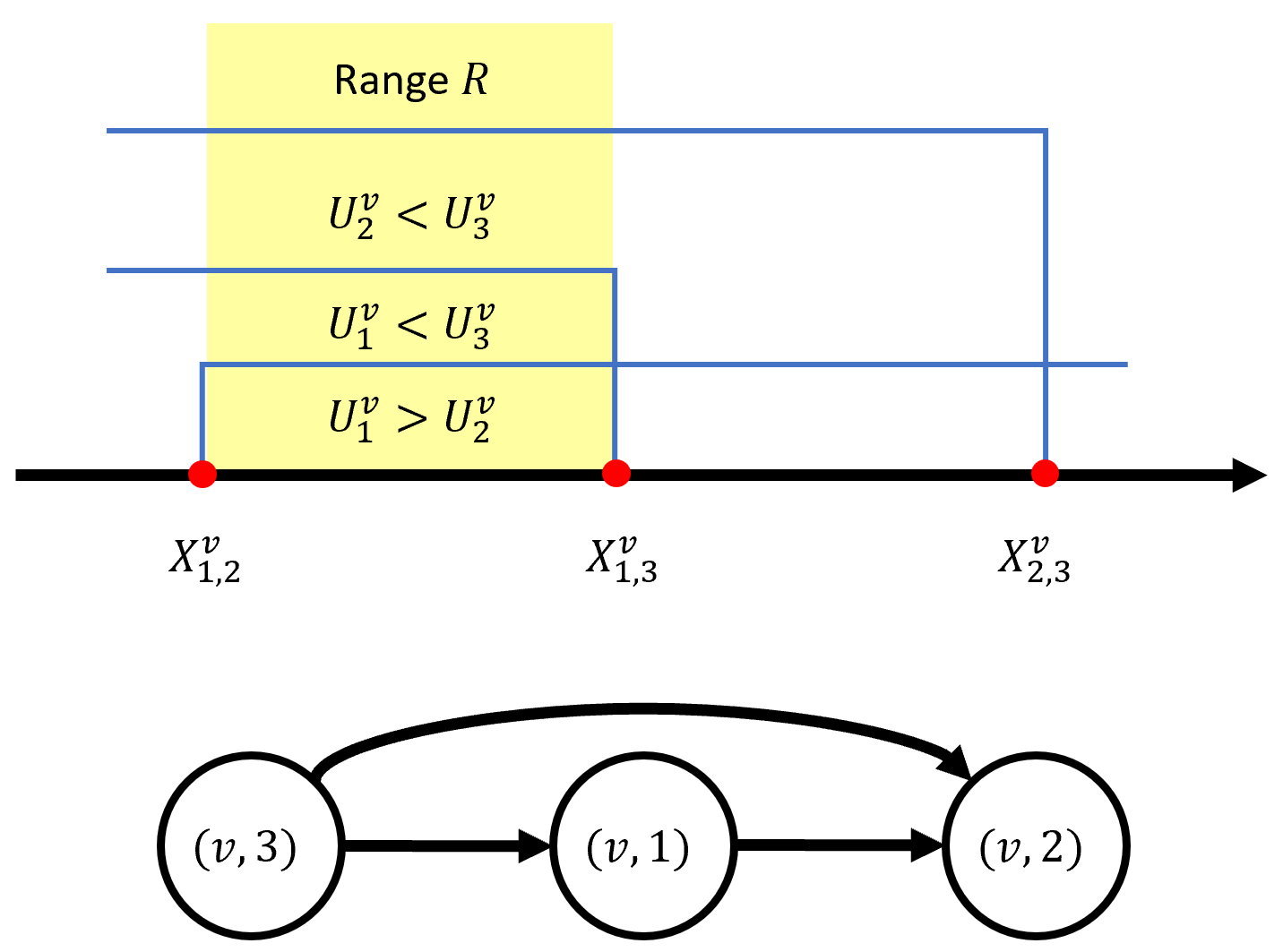}
    \caption{Axis of $X_0$. When $X_0$ lies within a specific range $R$, $G_v$ is constructed.}
    \label{fig:X0_axis}
\end{figure}
We may determine for all group pair $(v,i),(v,j)$, which group's individual utility is higher based on the location of $X_0$ with respect to $X_{i,j}^v$, illustrated in {\bf Figure \ref{fig:X0_axis}}.
On the axis of the value of $X_0$, in each node $v$, every pair of group $(v,i)$ and $(v,j)$ defines a point $X_{i,j}^v$, termed an event point. 
Sort all event points on the axis, each pair of adjacent points defines a range. 
Assume at the Nash equilibrium, the value of $X_0$ is within a specific range $R$, we can construct a graph $G_v$ representing the relationship between each group's individual utility in node $v$. We create a vertex for each group $(v,i)$, and a directed edge from $(v,i)$ to $(v,j)$ if $U_i^v \geq U_j^v$. This is a directed tournament graph with an edge between every pair of nodes. We next prove $G_v$ is acyclic.
\begin{lemma}
    The relationship graph $G_v$ is acyclic.
\end{lemma}
\begin{proof}
    Assume there is a directed cycle in $G_v$, $(v,i)\rightarrow(v,j)\rightarrow \cdots \rightarrow(v,i)$. This implies that when $X_0$ is in the current range, $U_i^v > U_j^v > \cdots > U_i^v\implies U_i^v>U_i^v$, which is impossible.
\end{proof}
We can then perform topological sort on $G_v$. Since $G_v$ is acyclic and there is an edge between every pair of nodes, there is one unique source vertex $(v,i^*)$, representing a group whose individual utility dominates all other groups in node $v$. Since in {\bf Case \ref{NetworkCase2}} we assume there is only one group participating in the Nash equilibrium, we have $\phi_{i^*}^v=1$ and $\phi_j^v=0, \forall j\neq i^*$.

There are $O(n^2)$ event points for each node $v$, and $O(mn^2)$ event points in total for the network.
This gives $O(mn^2)$ value ranges in total on the axis of $X_0$.
When $X_0$ is within any specific range $R$, we can determine the relationship graph $G_v$ for every node $v$.
Sorting all event points on the space of  $X_0$ gives all value ranges of $X_0$ in $O(mn^2 \log mn^2)$ steps. 
 For every range $R$, we calculate $\phi$ by performing $m$ topological sorts to find the source vertex $(i^*,v)$ and set $\phi_{i^*}^v=1$ for each $G_v$.
For each non-source vertex $(j,v)$, set $\phi_j^v=0$. 
Thus the entire vector $\phi$ can be computed in $O(mn^2)$ steps in total. 
With vector $\phi$ computed, we compute the final sizes $S$ using the convex program in {\bf Section \ref{sec:net_convex}} in polynomial number of steps.
Lastly we compute $X_0=\sum_{(u,j)} \frac{\overline{\kappa}_j^u \beta_0}{\gamma}(S_j^u - \phi_j^{u}) $ in $O(mn)$ steps and check if $X_0$ is within the current range $R$.
If yes, we have found a Nash equilibrium. If not, we move to test the next range of $X_0$.
{\bf Case \ref{NetworkCase2}} finishes in polynomial number of steps. We summarize the method in {\bf Algorithm \ref{alg:net2}} in {\bf Appendix \ref{appendix:algo}}.

Since we showed that the Nash equilibrium always exists in the beginning of {\bf Section \ref{sec:network_policy}},
we are guaranteed to obtain at least one Nash equilibrium from {\bf Case \ref{NetworkCase1} \& \ref{NetworkCase2}} in polynomial number of steps. In this abstract we ignore numerical precision errors and the stopping conditions of the convex program in line 7 of the algorithm. 
\begin{theorem}\label{thm:net_Nash_algo}
    The Nash equilibrium of contagion game with uniform network policy sets can be computed in polynomial time.
\end{theorem}

\subsection{Price of Anarchy}\label{sec:network_POA}
In this section we present results on the price of anarchy(POA) for contagion in the uniform network game model.

We start with a lower bound  of  group 1's individual utility $U_1^v$ at any node $v$.
Recall that $\beta_{i,j}(v) =\kappa_i \kappa_j \beta_0(v)$ with $1=\kappa_1>\kappa_2>\cdots>\kappa_n>0$, so group 1 has the highest $\beta$.
We show that for a fixed $\phi_j^{v'}, \forall j,v'\neq v$,  $U_1^v$ is lowest when $\phi_1^v=1$.  Let
$\phi= [\phi_j^v,\phi^{v'}_j]$ be the vector of population, initially, across all nodes and all policy classes where we
assume w.l.o.g. that $v$ is the  node with index $1$ and  where $\phi_j^{v'}, v' \neq v$ will be assumed to be fixed.
Furthermore,
Let $\phi_{END}=[1,0,\cdots,0, (\phi_j^{v'})_{j,v'\neq v}]^T$.
We show that $U_1^v(\phi_{END}) \leq U_1^v(\overline{\phi}),
\forall \overline{\phi}=[\overline{\phi}_1^{v}, \overline{\phi}_2^v\, \cdots, \overline{\phi}_n^v, (\phi_j^{v'})_{j,v'\neq v}]^T$.
Note that the individual utility 
$U_1^v = p_1^v  \cdot h^v(\overline{S}_1^v)$
is an increasing function of $\overline{S}_1^v$.
We  show the following.
\begin{lemma}\label{lm:S1_bar_lower_boundN1}
$\overline{S}_1^v(\phi_{END}) \leq \overline{S}_1^v(\overline{\phi}),
\forall \overline{\phi}$.
\end{lemma}
\begin{proof}
We first express $\overline{S}_1^v$ as:
$
\overline{S}_1^v = (1-\epsilon)e^{\overline{\kappa}_1^v X_0}
$.
We consider the change in $X_0$ at any point
specified by $\overline \phi$ in the direction of $\phi_{END}-\overline{\phi}$.
$X_0 = \sum_{(u,j)}  \frac{\overline{\kappa}_j^u \beta_0}{\gamma} \phi^u_j(\overline{S}_j^u - 1) = \sum_{(u,j)} \frac{\overline{\kappa}_j^u \beta_0}{\gamma} \phi^u_j ((1-\epsilon)
e^{\overline{\kappa}_j^u X_0}-1)
$.
Consider the slope of $X_0$ with respect to $\phi^v_1$ we get
$
\frac{\mathrm{d} X_0}{\mathrm{d} \phi^v_1} (1- \sum_{u,j} \frac{(\overline{\kappa}_j^u)^2 \beta_0}{\gamma}S^u_j ) \leq 0
$.
Since $1- \sum_{u,j} \frac{(\overline{\kappa}_j^u)^2 \beta_0}{\gamma}S^u_j >0$ we get the result that $\frac{\mathrm{d} \overline{S}_1^v}{\mathrm{d} \phi^v_1} < 0$.
\end{proof}
Repeating the above argument for all nodes we get the following result, where\\
$\phi_{END}=[\phi^v_{END}]_v^T, \phi^v_{END} = [1,0,\cdots,0, ]^T$.
\begin{lemma}
\label{lm:S1_bar_lower_boundN2}
$\overline{S}_1(\phi_{END}) \leq \overline{S}_1(\overline{\phi}),
\forall \overline{\phi}$.
\end{lemma}

Thus $U_1(\phi_{END})<U_1(\overline{\phi})$, for all $\overline{\phi}$ with $\overline{\phi}_1^v<1$.
Group 1's individual utility is lowest when its group size is 1, which is a lower bound of any Nash equilibrium.


When at $\phi_{END}$
the entire population of  every node is in group 1, there is no interaction with other groups. We determine a lower bound on $\overline{S}_1$.
Let $X_0 = \sum_{(u,j)}  \frac{\overline{\kappa}_j^u \beta_0}{\gamma} (S_j^u - \phi_j^u)$, $X_0<0$. 
For all groups $(v,i)$, 
$X_i^v = \sum_{(u,j)} \frac{\beta_{i,j}^{v,u}}{\gamma}(S_j^u - \phi_j^u) 
= \overline{\kappa}_i^v X_0$.
   Note that $\overline{\kappa}_i^v = \alpha_v \kappa_i$, for all group $(v,i)$
   in the uniform model.
We first find a lower bound for $X_0$ at $\phi_{END}$ where
$\phi_1^u =1$, $\forall u$.
$
X_0 = \sum_{u}  \frac{\overline{\kappa}_1^u \beta_0}{\gamma} (\overline{S}_1^u - 1) 
= \sum_{u}  \frac{\alpha_u \beta_0}{\gamma} (\hS_1^u - 1)
\geq -R_0 \sum_u \alpha_u
$, since $\kappa_1 =1$.
Thus
$
\overline{S}_1^v = (1-\epsilon)e^{\overline{\kappa}_1^v X_0}
\geq  (1-\epsilon)e^{\alpha_v (-R_0 \sum_u \alpha_u)}
$.
 Representing the social welfare at Nash equilibrium to be $N>0$, we get
 \[ N \geq  \sum_u \hp^u_1\phi_1^u \overline{S}^u_1  
 \geq \sum_u \hp^u_1   (1-\epsilon)e^{-\alpha_{max} \overline{R}_0 },
 \]
 where $\overline{R}_0 = R_0 \sum_u \alpha_u $
 and $\alpha_{max} = \max_u \alpha_u$.

The social optimum is
\[ OPT = \sum_u  \max_{\phi^u} \sum_i UG^u_i 
=\sum_u \max_{\phi^u} \sum_i  \hp^u_i S^u_i\]
Since for all $i$, $S^u_i\leq S^u_i(0) = (1 - \epsilon) \phi^u_i$,
\begin{align*}
OPT &\leq \sum_u \max_{\phi^u} \sum_i \hp^u_i(1-\epsilon)\phi^u_i 
\leq \sum_u \max_{\phi^u} \sum_i \hp^u_1(1-\epsilon)\phi^u_i
= \sum_u (1-\epsilon)\hp_1^u
\end{align*}
Assume $N$ to be fixed, we set up the following maximization program with variables $\{\hp_1^u\}_u$:
\begin{equation}
\label{prog:network_OPT}
\max_{\{\hp^u_1\}_u}\quad \sum_u (1-\epsilon)\hp^u_1 ; \  \  \mathrm{s.t.}
\quad N \geq \sum_u \frac{\hp^u_1}{B_1}
\mathrm{, where }B_1 = e^{\alpha_{max} \overline{R}_0 }/(1-\epsilon)
\end{equation}
The optimum value of {\bf Program \ref{prog:network_OPT}} is $(1-\epsilon)NB_1$.
The price of anarchy(POA) is bounded by the following
\begin{align*}
    POA = \frac{OPT}{N}
    \leq \frac{(1 - \epsilon) N B_1}{N} =(1 - \epsilon) B_1
   =  e^{\alpha_{max} \overline{R}_0} ,
\end{align*}
where $\overline{R}_0 = R_0 \sum_u \alpha_u = R_0 \omega$ and $\alpha_{max}=\max_u \alpha_u \leq 1$. The interaction factor $\omega$ can be $m$ but the interaction of a population at a node would typically be limited to a constant factor of the population at that node. 
We summarize the results as:
\begin{theorem}\label{thm:network_POA}
The price of anarchy for the uniform  network contagion game  is bounded above by $e^{ \overline{R}_0}= e^{\omega R_0}$, where $\omega$ is the interaction factor.
In the worst case this is bounded by $e^{mR_0}$ where $R_0$ is the maximum reproduction number of the contagion over all policy sets, and $m$ is the number of nodes in the network.
\end{theorem}


\bibliographystyle{plain}
\bibliography{SIRgame}
\pagebreak

\appendix
\centerline{\LARGE \bf Appendix}
\section{ Separable Policy Sets}\label{sec:sep}

In this section we consider another special case of the game, with no interaction between different policy groups. 
All the off-diagonal entries of the $\beta$ matrix are 0. 
For convenience, for all $i$ we denote $\beta_i = \beta_{i,i}$ with $\beta_0 = \beta_1 > \beta_2>\cdots >\beta_n$, as each group only has one $\beta$ parameter. 

An important observation is that in the separable policy sets model, $S_i(\infty)<\frac{\gamma}{\beta_i}$ for every group.
\begin{lemma}\label{lm:1_group_final_upper}
    $S_i(\infty)<\frac{\gamma}{\beta_i}$.
\end{lemma}
\begin{proof}
    $I_i(\infty)=0$, which means at a time $t\to\infty$ we have \begin{align*}
        \frac{\mathrm{d}I_i(t)}{\mathrm{d}t}=\beta_i S_i(t) I_i(t) -\gamma I_i(t)<0
        \implies S_i(t)<\frac{\gamma}{\beta_i}
    \end{align*}
    Since $\frac{\mathrm{d}S_i(t)}{\mathrm{d}t}$ is strictly non-positive, $S_i(\infty)<S_i(t)<\frac{\gamma}{\beta_i}$.
\end{proof}

\subsection{Computing the Final Size}\label{sec:sep_LB}
We first discuss how to find bounds on the final size. We omit the group index $i$ in this subsection for simplicity.\\
\noindent
{\bf{Bounds on Final Size:}}
We first derive the lower and upper bound of a group's final size, or its susceptible size at $T=\infty$, denoted by $S(\infty)$. 
It is known\cite{magal2016final} that
$S(\infty) = S(0) e^{\frac{\beta}{\gamma}[S(\infty)-\phi]}$.
Since there is no interaction with other groups, a group's final size $S(\infty)$ is only a function of its own $\beta$ and group size $\phi$.
Let $\overline{S}(\infty) = \frac{S(\infty)}{\phi}$, we get
$\overline{S}(\infty) = (1 - \epsilon)
e^{\frac{\beta\phi}{\gamma}[\overline{S}(\infty) - 1]}$.
Define function $g(\overline{S})=\overline{S} - (1 - \epsilon)
e^{\frac{\beta\phi}{\gamma}(\overline{S} - 1)},\ 0\leq \overline{S}\leq 1-\epsilon$. 
$\overline{S}(\infty)$ satisfies that $g(\overline{S}(\infty))=0$.
\begin{figure}
    \centering
    \includegraphics[width = 0.8\textwidth]{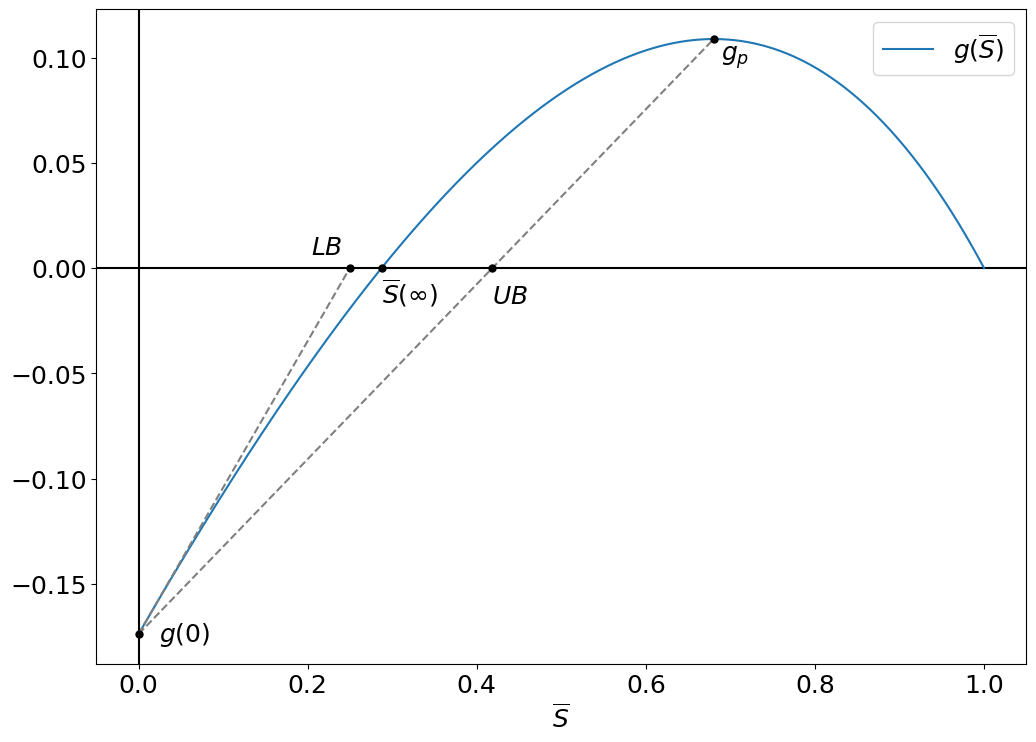}
    \caption{$g(\overline{S})$ and the lower bound $LB$ and upper bound $UB$ of $\overline{S}(\infty)$.}
    \label{fig:gS}
\end{figure}
\[
\frac{\mathrm{\mathrm{d}g}}{\mathrm{d}\overline{S}}=1-(1-\epsilon)
\frac{\beta\phi}{\gamma}e^{\frac{\beta\phi}{\gamma}(\overline{S}-1)}=0\implies
\overline{S}=\frac{\gamma}{\beta\phi}\ln(\frac{\gamma}{(1-\epsilon)\beta\phi}) + 1
\]
Note that $\frac{\mathrm{\mathrm{d}g}}{\mathrm{d}\overline{S}}>0$ when $\overline{S} < \frac{\gamma}{\beta\phi}\ln(\frac{\gamma}{(1-\epsilon)\beta\phi}) + 1$ and $\frac{\mathrm{\mathrm{d}g}}{\mathrm{d}\overline{S}}<0$ when $\overline{S} > \frac{\gamma}{\beta\phi}\ln(\frac{\gamma}{(1-\epsilon)\beta\phi}) + 1$. 
And $g(\overline{S})$ is concave since
$
\frac{\mathrm{\mathrm{d}^2g}}{\mathrm{d}\overline{S}^2}
=-(1-\epsilon)
\Big(\frac{\beta\phi}{\gamma}\Big)^2 e^{\frac{\beta\phi}{\gamma}(\overline{S}-1)} <0$.
Thus we get the peak point $g_p$.
\[g_p= \Bigg(\frac{\gamma}{\beta\phi}\ln(\frac{\gamma}{(1-\epsilon)\beta\phi}) + 1,
\frac{\gamma}{\beta\phi}[\ln(\frac{\gamma}{(1-\epsilon)\beta\phi})-1] +1 \Bigg)\]
Let $z=\frac{\beta\phi}{\gamma}$, $0<z\leq \frac{\beta}{\gamma}$, the peak $\frac{\gamma}{\beta\phi}[\ln(\frac{\gamma}{(1-\epsilon)\beta\phi})-1] +1$ is a function 
$p(z)=z[\ln(\frac{z}{1-\epsilon})-1]+1$
of $z$.
$p(z)$ has a minimum value of $\epsilon>0$ when $z=1-\epsilon$, the peak is above 0.
Connecting $(0,g(0))$ and $g_p$, we get the intersection $(UB,0)$ on x-axis as the upper bound of $\overline{S}(\infty)$
\[
UB=\frac{(1-\epsilon)\Big(
\frac{\beta\phi}{\gamma} - \ln(\frac{(1-\epsilon)\beta\phi}{\gamma})
\Big)}
{\frac{(1-\epsilon)\beta\phi}{\gamma} +
e^{\frac{\beta\phi}{\gamma}}
\Big(\frac{\beta\phi}{\gamma}-1-\ln(\frac{(1-\epsilon)\beta\phi}{\gamma})
\Big)}
\]
We extend the tangent at $(0, g(0))$ to intersect the x-axis, with the slope
$\left.\frac{\mathrm{d}g}{\mathrm{d}\overline{S}}\right|_{\overline{S}=0}=1-(1-\epsilon)
\frac{\beta\phi }{\gamma}e^{-\frac{\beta\phi}{\gamma}}$.
The intersection $(LB,0)$ is the lower bound of $\overline{S}(\infty)$.
\begin{equation}\label{eq:LB}
    LB=\frac{1}{\frac{e^{\frac{\beta\phi}{\gamma}}}{1-\epsilon} -
\frac{\beta\phi}{\gamma}}
\end{equation}
Using the above analysis we get the following result. Note that $UB,LB$ are functions of $\phi, \beta, \gamma, \epsilon$. 
\begin{lemma}\label{lm:1_group_UB_LB}
$\overline{S}(\infty)$, the size at endemicity, is bounded above and below as:
\[
\frac{(1-\epsilon)\Big(
\frac{\beta\phi}{\gamma} - \ln(\frac{(1-\epsilon)\beta\phi}{\gamma})
\Big)}
{\frac{(1-\epsilon)\beta\phi}{\gamma} +
e^{\frac{\beta\phi}{\gamma}}
\Big(\frac{\beta\phi}{\gamma}-1-\ln(\frac{(1-\epsilon)\beta\phi}{\gamma})
\Big)}
\geq \overline{S}(\infty) \geq 
\frac{1}{\frac{e^{\frac{\beta\phi}{\gamma}}}{1-\epsilon} -
\frac{\beta\phi}{\gamma}}\]
\end{lemma}
\noindent
{\bf{Computing $S_i(\infty)$:}} $g(\overline{S})$ is an increasing function when $LB\leq \overline{S}\leq UB$, we may use binary search to determine the numerical value of $\overline{S}(\infty)$, and that of $S(\infty)=\phi\cdot\overline{S}$. 
For each group $i$, we define a function $FinalSize_i(\phi_i)$ that takes in the group size $\phi_i$ and returns the corresponding final size $S_i$ by the binary search between $LB$ and $UB$ described above.
For the error in $\overline{S}(\infty)$ to be bounded by $\delta$, it takes $O(\log \frac{1}{\delta})$ iterations.

\begin{theorem}\label{thm:sep_binary}
    The final size  of each group in contagion game with separable policy sets can be approximated to within an additive error of $\delta$ in $O(\log \frac{1}{\delta})$.
\end{theorem}

\section{Proofs}\label{appendix:proofs}

\begin{proof}[Proof of {\bf Lemma \ref{lm:inverse_function}}]
\[
f(N) = \frac{\ln\frac{N}{1-\epsilon}}{N-1},\quad 0\leq N\leq 1-\epsilon
\]
$f(0) = +\infty$ and $f(1-\epsilon)=0$. 
The derivative
\[
\frac{\mathrm{d}f}{\mathrm{d}N} = \frac{1 - \frac{1}{N} - \ln{\frac{N}{1-\epsilon}}}{(N-1)^2}
\]
The denominator $>0$. Denote the numerator by
\[
g(N)=1 - \frac{1}{N} - \ln{\frac{N}{1-\epsilon}},\quad 0\leq N\leq 1-\epsilon
\]
$g(1-\epsilon) =1-\frac{1}{1-\epsilon}<0$, and the derivative
\[
\frac{\mathrm{d}g}{\mathrm{d}N} = \frac{1}{N^2} - \frac{1}{N}
=\frac{1}{N}\Big(\frac{1}{N} - 1\Big)>0
\]
Thus $g(N)\leq g(1-\epsilon)<0$, the numerator $<0$, $f(N)$ decreases monotonically from $0$ to $1-\epsilon$. 
\end{proof}

\begin{proof}
[Proof of {\bf Lemma \ref{lm:n_group_final_upper}}]
    Since $S_i$ are the final sizes and $\frac{\mathrm{d}S_i(t)}{\mathrm{d}t}<0$ strictly for any time $t<\infty$, there must exist a time $t$ with $\frac{\mathrm{d}I_i(t)}{\mathrm{d}t}<0$ and $S_i<S_i(t)$ for all $i$.
    \begin{align}
        \frac{\mathrm{d}I_i(t)}{\mathrm{d}t}=
        &S_i(t) \sum_{j=1}^n \beta_{i,j}I_j(t)
        - \gamma I_i(t)<0
        \nonumber
        \implies S_i(t) \sum_{j=1}^n \kappa_i \kappa_j \beta_0 I_j(t)
        - \gamma I_i(t)<0
        \nonumber\\
        \implies &\frac{\kappa_i\beta_0}{\gamma} S_i(t) \sum_{j=1}^n \kappa_j I_j(t)
        <I_i(t) \label{ineq03}
    \end{align}
    Multiplying both sides of inequality (\ref{ineq03}) by $\kappa_i >0$ and summing
    over $i$ we get:
    \begin{align*}
    \sum_{i =1}^n
    \frac{\kappa_i^2\beta_0}{\gamma} S_i(t) \sum_{j=1}^n \kappa_j I_j(t)
    <\sum_{i =1}^n \kappa_i I_i(t)
    \implies \sum_{i =1}^n \frac{\kappa_i^2\beta_0}{\gamma} S_i(t) <1
    \end{align*}
\end{proof}
\begin{proof}
[Proof of {\bf Lemma \ref{lm:f_i_concave}}]
    The Hessian of $f_i$
    \begin{align*}
        H_{f_i}&=
        \begin{bmatrix}
        -S_i(0)\frac{\beta_{i,1}\beta_{i,1}}{\gamma^2}e^{X_i} 
        & -S_i(0)\frac{\beta_{i,1}\beta_{i,2}}{\gamma^2}e^{X_i}
        & \cdots
        & -S_i(0)\frac{\beta_{i,1}\beta_{i,n}}{\gamma^2}e^{X_i}\\
        -S_i(0)\frac{\beta_{i,2}\beta_{i,1}}{\gamma^2}e^{X_i} 
        & -S_i(0)\frac{\beta_{i,2}\beta_{i,2}}{\gamma^2}e^{X_i}
        & \cdots
        & -S_i(0)\frac{\beta_{i,2}\beta_{i,n}}{\gamma^2}e^{X_i}\\
        \vdots & &\ddots &\vdots\\
        -S_i(0)\frac{\beta_{i,n}\beta_{i,1}}{\gamma^2}e^{X_i} 
        & -S_i(0)\frac{\beta_{i,n}\beta_{i,2}}{\gamma^2}e^{X_i}
        & \cdots
        & -S_i(0)\frac{\beta_{i,n}\beta_{i,n}}{\gamma^2}e^{X_i}\\
        \end{bmatrix}\\
        &=-(1-\epsilon)\phi_i\frac{\kappa_i^2 \beta_0^2}{\gamma^2}e^{X_i} \cdot
        \vec{\kappa}\cdot\vec{\kappa}^T\\
        \text{Thus }&\forall x \in \mathbb{R}^n,\quad
        x^T H_{f_i} x = -(1-\epsilon)\phi_i\frac{\kappa_i^2 \beta_0^2}{\gamma^2}e^{X_i}\cdot 
        (x^T \vec{\kappa})\cdot(\vec{\kappa}^T x)\leq 0
    \end{align*}
    $H_{f_i}$ is negative semi-definite, $f_i$ is concave.
\end{proof}

\begin{proof}[Proof of {\bf Lemma \ref{lm:f_i_Delta_infeasible}}]
    Assume $\Delta$ is a very small deviation,
    \[
        f_i(F^*+\Delta) 
        = f_i(F^*) + J_{f_i}^T \cdot \Delta
        =J_{f_i}^T \cdot \Delta
    \]
    \begin{align*}
        J_{f_i}^T\cdot \Delta
        &=\Delta_i + \sum_{j=1}^n (-\frac{\beta_{i,j}}{\gamma}S_i\Delta_j)\\
        \kappa_i \cdot J_{f_i}^T\cdot \Delta
        &=\kappa_i \Delta_i + 
        \frac{\kappa_i^2\beta_0}{\gamma}S_i
        \sum_{j=1}^n (- \kappa_j \Delta_j)\\
    \end{align*}
    \begin{align}
        \sum_{i \in \Delta_-}
        \kappa_i \cdot J_{f_i}^T\cdot \Delta
        &=\sum_{i \in \Delta_-}\kappa_i \Delta_i + 
        \sum_{i \in \Delta_-}
        \frac{\kappa_i^2\beta_0}{\gamma}S_i
        \sum_{j=1}^n (- \kappa_j \Delta_j)\label{ineq04}
    \end{align}
    Let $\kappa_{min}=\min_{i \in \Delta_+}\kappa_i$.
    \begin{align*}
        \sum_{j\in \Delta_+}\kappa_j\Delta_j &\geq
        \kappa_{min}\sum_{j\in \Delta_+}\Delta_j\\
        \sum_{j=1}^n \kappa_j\Delta_j &=
        \sum_{j\in \Delta_-} \kappa_j\Delta_j +
        \sum_{j\in \Delta_+} \kappa_j\Delta_j\\
        &\geq \sum_{j\in \Delta_-} \kappa_j\Delta_j + 
        \kappa_{min}\sum_{j\in \Delta_+}\Delta_j\\
        &= \sum_{j\in \Delta_-} \kappa_j\Delta_j - 
        \kappa_{min}\sum_{j\in \Delta_-}\Delta_j
    \end{align*}
    \begin{align}
        \therefore \sum_{j=1}^n (-\kappa_j\Delta_j) &\leq
        - \sum_{j\in \Delta_-} \kappa_j\Delta_j + 
        \kappa_{min}\sum_{j\in \Delta_-}\Delta_j\label{ineq05}
    \end{align}
    From $(\ref{ineq04}) \& (\ref{ineq05})$,
    \begin{align*}
        \sum_{i \in \Delta_-}
        \kappa_i \cdot J_{f_i}^T\cdot \Delta 
        &\leq
        \sum_{i \in \Delta_-}\kappa_i \Delta_i + 
        \sum_{i \in \Delta_-}
        \frac{\kappa_i^2\beta_0}{\gamma}S_i
        \Big(
        - \sum_{j\in \Delta_-} \kappa_j\Delta_j + 
        \kappa_{min}\sum_{j\in \Delta_-}\Delta_j
        \Big)\\
        &= \sum_{i \in \Delta_-}
        \Big[\kappa_i \Delta_i
        \Big(
        1 - \sum_{j \in \Delta_-}\frac{\kappa_j^2\beta_0}{\gamma}S_j
        \Big)\Big]
        + \sum_{i\in\Delta_-} \frac{\kappa_i^2\beta_0}{\gamma}S_i \cdot
        \kappa_{min}\sum_{j\in \Delta_-}\Delta_j\\
        &\leq \sum_{i\in\Delta_-} \frac{\kappa_i^2\beta_0}{\gamma}S_i \cdot
        \kappa_{min}\sum_{j\in \Delta_-}\Delta_j
        \text{, by \textbf{Lemma \ref{lm:n_group_final_upper}} and}
        \sum_{i \in \Delta_-}\kappa_i \Delta_i<0\\
        & <0
    \end{align*}
    Since $\sum_{i \in \Delta_-}\kappa_i \cdot J_{f_i}^T\cdot \Delta<0$, there exists $i\in \Delta_-$ such that
        \[\kappa_i \cdot J_{f_i}^T\cdot \Delta<0
        \implies f_i(F^*+\Delta) = J_{f_i}^T\cdot \Delta<0\]
\end{proof}

\begin{proof}[Proof of {\bf Lemma \ref{lm:decomp_delta_bound}}]
    Assume the algorithm is testing whether $\phi_i=1, \phi_j=0,\forall j\neq i$ is a Nash equilibrium. The condition is
    \[
    \frac{\overline{p}_i\overline{S}_i}
    {\overline{p}_j\overline{S}_j}\geq 1
    ,\forall j
    \]
    The numerical calculation for $\overline{S}_i$ may introduce an error $\delta$. The estimate of $\overline{S}_i$ is $\widetilde{S}_i=\overline{S}_i \pm \delta$. For convenience, in the condition, for all $j$, $\overline{S}_j$ is estimated by 
    \[
    \widetilde{S}_j = (1-\epsilon)e^{\frac{\beta{j,i}}{\gamma}(\widetilde{S}_i - 1)},\]
    including when $j=i$. Since all parameters are provided in the form of $n_1/n_2$ with $n_1,n_2\leq n_0$, the smallest step size is $\frac{1}{n_0^2}$. To avoid  computational error after rounding, we require the error introduced by $\delta$ to be bounded by $1/(4 n_0^2)$.
    We need to make sure the following 2 cases.
    \begin{enumerate}[label=(\roman*)]
    \item 
    \[\frac{\overline{p}_i\overline{S}_i}
    {\overline{p}_j\overline{S}_j}\geq 1\implies
    \frac{\overline{p}_i\widetilde{S}_i}
    {\overline{p}_j\widetilde{S}_j}\geq 1 - \frac{1}{4 n_0^2}\]
    From $\frac{\overline{p}_i\overline{S}_i}
    {\overline{p}_j\overline{S}_j}\geq 1$ we get
    $\frac{\overline{p}_i}{\overline{p}_j}
    \geq
    \frac{\overline{S}_j}{\overline{S}_i}$. Thus we have 
    \begin{align*}
        \frac{\overline{p}_i\widetilde{S}_i}
        {\overline{p}_j\widetilde{S}_j}\geq
        \frac{\overline{S}_j\widetilde{S}_i}{\overline{S}_i\widetilde{S}_j}
    \end{align*}
    It suffices to show
    \begin{align*}
        \frac{\overline{S}_j\widetilde{S}_i}{\overline{S}_i\widetilde{S}_j}\geq
        1 - \frac{1}{4 n_0^2} &\impliedby\\
        \frac{exp(\kappa_j\kappa_i R_0 (\overline{S}_i -1))
        exp(\kappa_i\kappa_i R_0 (\overline{S}_i \pm \delta-1))}
        {exp(\kappa_i\kappa_i R_0 (\overline{S}_i -1))
        exp(\kappa_j\kappa_i R_0 (\overline{S}_i \pm \delta-1))}
        \geq 1 - \frac{1}{4 n_0^2} &\impliedby\\
        \kappa_j\kappa_i R_0 (\overline{S}_i -1) +
        \kappa_i\kappa_i R_0 (\overline{S}_i \pm \delta-1)&\\
        - \kappa_i\kappa_i R_0 (\overline{S}_i -1) - 
        \kappa_j\kappa_i R_0 (\overline{S}_i \pm \delta-1) 
        \geq \ln(1 - \frac{1}{4 n_0^2}) &\impliedby\\
        \pm(\kappa_i - \kappa_j) \kappa_i R_0 \delta 
        \geq \ln(1 - \frac{1}{4 n_0^2}) &\impliedby 
    \end{align*}
    $\ln(1 - \frac{1}{4 n_0^2})\approx - \frac{1}{4 n_0^2}<0$. We need to only look at when the left-hand side is negative, where we need
    \begin{align*}
        \delta \leq 
        \frac{\frac{1}{4n_0^2}}{|\kappa_i-\kappa_j|\kappa_i R_0}
    \end{align*}
    Since $0<\kappa_j\leq 1,\forall j$ and $R_0=\beta_0/\gamma$, both $\beta_0,\gamma$ are specified by the form $n_1/n_2$, we get
    \[
    \frac{\frac{1}{4 n_0^2}}{|\kappa_i-\kappa_j|\kappa_i R_0}
    \geq \frac{1}{4 n_0^4}
    \]
    As long as we choose $\delta$ to be $\frac{1}{4 n_0^4}$, the correctness of this cases is guaranteed.
    \item 
    \[\frac{\overline{p}_i\overline{S}_i}
    {\overline{p}_j\overline{S}_j}\leq 1\implies
    \frac{\overline{p}_i\widetilde{S}_i}
    {\overline{p}_j\widetilde{S}_j}\leq 1 + \frac{1}{4 n_0^4}\]
    This is a symmetric case and we get the exactly same bound for $\delta$.
    \end{enumerate}
    This choice of $\delta$ guarantees that after numerical rounding, the algorithm is correct.
\end{proof}

\begin{proof}[Proof of {\bf Lemma \ref{lm:S1_bar_lower_bound}}]
Without loss of generality, we assume $\overline{\phi}_i>0, \forall i=2,3,\cdots,n$, for otherwise we can remove the group from the system in this analysis.
For every point $\overline{\phi}$ with strictly positive components, we define a straight line segment $\phi$ in the domain, such that both $\overline{\phi}$ and $\phi_{END}$ are on it. 
\[
\phi(\theta) = \theta\cdot \phi_{END} + (1-\theta)\cdot \phi_{START}, 0<\theta\leq 1
\]
where $\phi_{START}=[0, r_2, r_3,\cdots ,r_n]^T$, with
\[r_i = \frac{\overline{\phi}_i}{\sum_{j=2}^n \overline{\phi}_j}, \forall i=2,3,\cdots,n
\]
Thus $\phi_{END} = \phi(1)$ and $\overline{\phi} = \phi(\theta)$ for some $\theta < 1$. On this line segment $\phi$,we show that
\[
\frac{\mathrm{d}\overline{S}_1}{\mathrm{d}\theta}<0
\]
By the definition of $\phi(\theta)$,
$\frac{\mathrm{d}\phi_1}{\mathrm{d}\theta} = 1$, denote
\[
D_i=\frac{\mathrm{d}\overline{S}_i}{\mathrm{d}\phi_i}
\]
we need to show $D_1<0$.
Let $r_1=-1$, we get
\[
\frac{\mathrm{d}\overline{\phi}_i}{\mathrm{d}\overline{\phi}_j} = 
\frac{r_i}{r_j}, \forall i,j
\]
Denote $X_i = \sum_j \frac{\beta_{i,j}\phi_j}{\gamma}(\overline{S}_j-1)$, we have $\overline{S}_i = (1-\epsilon)e^{X_i}$. For all $i$,
\begin{align*}
    D_i &= \frac{\mathrm{d}\overline{S}_i}{\mathrm{d}\phi_i}
    =(1-\epsilon)e^{X_i}\frac{\mathrm{d}X_i}{\mathrm{d}\phi_i}\\
    &=\overline{S}_i\sum_j\Big(
    \frac{\beta_{i,j}}{\gamma}(\overline{S}_j-1)
    \frac{\mathrm{d}\phi_j}{\mathrm{d}\phi_i} +
    \frac{\beta_{i,j}\phi_j}{\gamma} 
    \frac{\mathrm{d}\overline{S}_j}{\mathrm{d}\phi_i}
    \Big)\\
    &=\overline{S}_i\sum_j\Big(
    \frac{\beta_{i,j}}{\gamma}(\overline{S}_j-1)
    \frac{\mathrm{d}\phi_j}{\mathrm{d}\phi_i} +
    \frac{\beta_{i,j}\phi_j}{\gamma} 
    \frac{\mathrm{d}\overline{S}_j}{\mathrm{d}\phi_j}
    \frac{\mathrm{d}\phi_j}{\mathrm{d}\phi_i}
    \Big)\\
    D_i &= \overline{S}_i\sum_j\Big(\frac{\kappa_i\kappa_j\beta_0}{\gamma}(\overline{S}_j-1)\frac{r_j}{r_i} + 
    \frac{\kappa_i\kappa_j\beta_0}{\gamma}\phi_j
    \frac{r_j}{r_i}D_j
    \Big)\\
    \frac{\gamma}{\kappa_i\beta_0\overline{S}_i}r_i D_i 
    &= \sum_j \kappa_j \phi_j r_j D_j + 
    \sum_j \kappa_j r_j(\overline{S}_j-1)
\end{align*}
Let $\hat{D}_i = r_i D_i$,
\[
\frac{\gamma}{\kappa_i\beta_0\overline{S}_i}\hat{D}_i -
\sum_j \kappa_j \phi_j \hat{D}_j = \sum_j \kappa_j r_j(\overline{S}_j-1),\forall i
\]
This is equivalent to the system of linear equations $\hat{A}\cdot \hat{D} = b$, where
\begin{align*}
    b&=\Big(
    \sum_j \kappa_j r_j(\overline{S}_j-1)\Big)\cdot \textbf{1},\\
    \hat{D} &= [\hat{D}_1,\hat{D}_2,\cdots,\hat{D}_n]^T,\\
    \hat{A} &= A+u\cdot v^T,\text{ where}\\
    A&=\frac{\gamma}{\beta_0}\cdot 
    diag(\frac{1}{\kappa_1 \overline{S}_1},
    \frac{1}{\kappa_2 \overline{S}_2},\cdots,
    \frac{1}{\kappa_n \overline{S}_n}),\\
    u&=\mathbf{-1},\\
    v^T&=[\kappa_1\phi_2,\kappa_1\phi_2,\cdots,\kappa_n\phi_n]
\end{align*}
By \textbf{Sherman-Morrison formula},
\[
    \hat{A}^{-1} = A^{-1} - \frac{A^{-1}uv^T A^{-1}}
    {1 + v^T A^{-1} u}
\]
\begin{align*}
    \hat{D}_1 &= (\hat{A}^{-1}b)[1]\\
    &= \frac{\beta_0}{\gamma}\kappa_1\overline{S}_1
    \sum_j \kappa_j r_j(\overline{S}_j-1) + 
    \frac{
    (\frac{\beta_0}{\gamma})^2
    \Big(
    \sum_j \kappa_1 \overline{S}_1 \kappa_j^2 \overline{S}_j \phi_j
    \Big)
    \Big(
    \sum_j \kappa_j r_j (\overline{S}_j-1)
    \Big)}
    {1 - \frac{\beta_0}{\gamma}
    \sum_j \kappa_j^2 \overline{S}_j\phi_j}
\end{align*}
The full inversion of $\hat{A}$ can be found in \textbf{Appendix \ref{inversion}}. Since $r_1=-1$, to show $D_1<0$ is to show $\hat{D}_1>0$.

    \begin{enumerate}[label=(\roman*)]
        \item $\sum_j \kappa_j r_j(\overline{S}_j-1)>0$\\
        Denote
        \[
        X_0 = \sum_j\frac{\kappa_j\beta_0\phi_j}{\gamma}(\overline{S}_j-1),\]
        we get for all $j$,
        \begin{align*}
            \overline{S}_j &= (1-\epsilon)exp(\kappa_j X_0)\\
            \frac{\overline{S}_1}{\overline{S}_j} &= (e^{X_0})^{(\kappa_1-\kappa_j)}, \forall j\neq 1
        \end{align*}
        Since $X_0<0$, $0<\kappa_1-\kappa_j<1$,
        \begin{align*}
            &\frac{\overline{S}_1}{\overline{S}_j}<1\implies \overline{S}_1 < \overline{S}_j
            \implies  \overline{S}_1 - 1 < \overline{S}_j - 1 <0\implies\\
            &\kappa_j(\overline{S}_j-1) > \kappa_1(\overline{S}_1-1)\implies\\
            &\sum_{j\neq 1} r_j\kappa_j(\overline{S}_j-1) >
            (\sum_{j\neq 1}r_j) \kappa_1(\overline{S}_1-1) = 
            \kappa_1(\overline{S}_1-1)\implies \\
            & \sum_{j\neq 1} r_j\kappa_j(\overline{S}_j-1) -
            \kappa_1(\overline{S}_1-1) >0
        \end{align*}
        Since $r_1=-1$, 
        \begin{align*}
            \sum_j r_j\kappa_j(\overline{S}_j-1) =
            \sum_{j\neq 1} r_j\kappa_j(\overline{S}_j-1) -
            \kappa_1(\overline{S}_1-1)>0
        \end{align*}
        \item $\frac{\hat{D}_1}{\sum_j r_j\kappa_j(\overline{S}_j-1)}>0$
        \begin{align*}
            \frac{\hat{D}_1}{\sum_j r_j\kappa_j(\overline{S}_j-1)}=
            \frac{\beta_0}{\gamma}\kappa_1\overline{S}_1
            \Big[
            1 +
            \frac{\frac{\beta_0}{\gamma}
            \sum_j \kappa_j^2 \overline{S}_j \phi_j}
            {1 - 
            \frac{\beta_0}{\gamma} \sum_j \kappa_j^2 \overline{S}_j \phi_j}
            \Big]
        \end{align*}
        \begin{align*}
            &\frac{\beta_0}{\gamma} \sum_j \kappa_j^2 \overline{S}_j \phi_j =
            \sum_j \frac{\kappa_j^2\beta_0}{\gamma}S_j <1
            \textbf{ by Lemma \ref{lm:n_group_final_upper}}\\
            &\implies \frac{\hat{D}_1}{\sum_j r_j\kappa_j(\overline{S}_j-1)}>0
        \end{align*}
    \end{enumerate}

Therefore $D_1 = \frac{\mathrm{d}S_1/\phi_1}{\mathrm{d}\phi_1}<0$,
 $\overline{S}_1(\phi_{END}) \leq \overline{S}_1(\overline{\phi}),
\quad \forall \overline{\phi}$.
\end{proof}

\section{Algorithms}
\label{appendix:algo}
\begin{algorithm}
\caption{Equilibrium Computation for Case \ref{NetworkCase1}}
\label{alg:net1}
\begin{algorithmic}[1]
\FOR {every $(v,i)(v,j)$ pair}
    \STATE $X_0 \gets X_{i,j}^v$
    \FOR {every group $l$ in every node $u$} 
        \STATE $U_l^u \gets p_l^u (1-\epsilon) e^{\overline{\kappa}_l^u X_0}$
    \ENDFOR

    \FOR{every group $(v,l)$ in node $v$}
        \IF{$U_l^v > U_i^v \textbf{ or } U_l^v > U_j^v$}
            \STATE Skip to the next  $(v,i)(v,j)$ pair in line 1
        \ENDIF
    \ENDFOR
    \FOR{every node $u\neq v$}
        \STATE $i^* \gets \argmax_i U_i^u$
        \STATE $\phi_{i^*}^u \gets 1;\quad\phi_{j}^u \gets 0,\ \forall j\neq i^*$
    \ENDFOR
    \STATE $\phi_l^v \gets 0, \forall l\neq i,j$
    \STATE Solve linear system {\bf (\ref{eq:phi_ij})} to calculate $\phi_i^v,\phi_j^v$
    \IF{$\phi_i^v,\phi_j^v \geq 0$}
        \STATE Equilibrium found, return vector $\phi$
    \ENDIF

\ENDFOR
\end{algorithmic}
\end{algorithm}

\begin{algorithm}
\caption{Equilibrium Computation for Case \ref{NetworkCase2}}
\label{alg:net2}
\begin{algorithmic}[1]
\STATE Calculate and sort $X_{i,j}^v$ for every $(v,i)(v,j)$ pair in every node $v$, to get every range of $X_0$
\FOR {every range $R$ of $X_0$}
    \FOR{every node $v$}
        \STATE Construct the relationship graph $G_v$ for every node $v$
        \STATE Perform topological sort on $G_v$ to determine the source group node representing group $(v,i^*)$
        \STATE $\phi_{i^*}^v \gets 1;\quad \phi_j^v \gets 0 \ \forall j\neq i^*$
    \ENDFOR
    \STATE With all $\phi$, compute the final sizes using the convex program in {\bf Section \ref{sec:net_convex}}
    \STATE With all $S$ and $\phi$, compute $X_0$
    \IF{$X_0$ in the current range $R$}
        \STATE Equilibrium found, return vector $\phi$
    \ENDIF
\ENDFOR
\end{algorithmic}
\end{algorithm}
\pagebreak
\section{Inversion of $\hat{A}$}\label{inversion}
\begin{align*}
    \hat{A} &= A+u\cdot v^T,\text{ where}\\
    A&=\frac{\gamma}{\beta_0}\cdot 
    diag(\frac{1}{\kappa_1 \overline{S}_1},
    \frac{1}{\kappa_2 \overline{S}_2},\cdots,
    \frac{1}{\kappa_n \overline{S}_n}),\\
    u&=\mathbf{-1},\\
    v^T&=[\kappa_1\phi_2,\kappa_1\phi_2,\cdots,\kappa_n\phi_n]
\end{align*}
By \textbf{Sherman-Morrison formula},
\[
    \hat{A}^{-1} = A^{-1} - \frac{A^{-1}uv^T A^{-1}}
    {1 + v^T A^{-1} u}
\]
\begin{enumerate}
    \item $A^{-1}$
    \[
    A^{-1} = \frac{\beta_0}{\gamma}\cdot 
    \begin{bmatrix}
        &\kappa_1\overline{S}_1 &0 &\cdots &0\\
        &0 &\kappa_2\overline{S}_2 &\cdots &0\\
        & & &\ddots &\\
        &0 &\cdots &0 &\kappa_n\overline{S}_n
    \end{bmatrix}
    \]

    \item $v^TA^{-1}$
    \begin{align*}
        v^TA^{-1} &= [\kappa_1\phi_1,\kappa_2\phi_2,\cdots,\kappa_n\phi_n]
        \cdot \frac{\beta_0}{\gamma}\cdot
        \begin{bmatrix}
            &\kappa_1\overline{S}_1 &0 &\cdots &0\\
            &0 &\kappa_2\overline{S}_2 &\cdots &0\\
            & & &\ddots &\\
            &0 &\cdots &0 &\kappa_n\overline{S}_n
        \end{bmatrix}\\
        &=\frac{\beta_0}{\gamma}
        [\kappa_1^2\overline{S}_1\phi_1,
        \kappa_2^2\overline{S}_2\phi_2,\cdots,
        \kappa_n^2\overline{S}_n\phi_n]
    \end{align*}

    \item $v^T A^{-1} u$
    \begin{align*}
        v^T A^{-1} u =\frac{\beta_0}{\gamma}
        [\kappa_1^2\overline{S}_1\phi_1,
        \kappa_2^2\overline{S}_2\phi_2,\cdots,
        \kappa_n^2\overline{S}_n\phi_n]\cdot
        \begin{bmatrix}
            -1\\
            -1\\
            \vdots\\
            -1
        \end{bmatrix}
        = -\frac{\beta_0}{\gamma}
        \sum_{j=1}^n \kappa_j^2 \overline{S}_j \phi_j
    \end{align*}

    \item $A^{-1} u$
    \begin{align*}
        A^{-1} u = \frac{\beta_0}{\gamma}\cdot 
        \begin{bmatrix}
            &\kappa_1\overline{S}_1 &0 &\cdots &0\\
            &0 &\kappa_2\overline{S}_2 &\cdots &0\\
            & & &\ddots &\\
            &0 &\cdots &0 &\kappa_n\overline{S}_n
        \end{bmatrix}\cdot
        \begin{bmatrix}
            -1\\
            -1\\
            \vdots\\
            -1
        \end{bmatrix}
        = -\frac{\beta_0}{\gamma}\cdot
        \begin{bmatrix}
            \kappa_1 \overline{S}_1\\
            \kappa_2 \overline{S}_2\\
            \vdots\\
            \kappa_n \overline{S}_n
        \end{bmatrix}
    \end{align*}

    \item $(A^{-1}u)(v^T A^{-1})$
    \begin{align*}
        (A^{-1}u)(v^T A^{-1}) &=
        -\frac{\beta_0}{\gamma}\cdot
        \begin{bmatrix}
            \kappa_1 \overline{S}_1\\
            \kappa_2 \overline{S}_2\\
            \vdots\\
            \kappa_n \overline{S}_n
        \end{bmatrix}
        \cdot
        \frac{\beta_0}{\gamma}
        [\kappa_1^2\overline{S}_1\phi_1,
        \kappa_2^2\overline{S}_2\phi_2,\cdots,
        \kappa_n^2\overline{S}_n\phi_n]\\
        &= - \Big(\frac{\beta_0}{\gamma}\Big)^2\cdot M_{n\times n}
    \end{align*}
    where
    \[
    M_{i,j} = \kappa_i\overline{S}_i \kappa_j^2 \overline{S}_j \phi_j,
    \quad \forall i,j \in 1,\cdots,n
    \]

    \item $\hat{A}^{-1}$
    \begin{align*}
        \hat{A}^{-1} =A^{-1} - \frac{A^{-1}uv^T A^{-1}}{1 + v^T A^{-1} u}
        =\frac{A^{-1} +(\frac{\beta_0}{\gamma})^2\cdot M}
        {1 - 
        \frac{\beta_0}{\gamma}\sum_{j=1}^n \kappa_j^2 \overline{S}_j \phi_j}
    \end{align*}
\end{enumerate}





\end{document}